\documentclass[a4paper,11pt]{article}
\usepackage{amsfonts,amssymb,amsmath,amsthm,dsfont,xspace}

\usepackage{fullpage,subfigure}
\usepackage{graphicx}
\makeatletter
\let\NAT@parse\undefined
\makeatother
\usepackage[sort&compress, numbers]{natbib}

\graphicspath{{./},{images/}}
\usepackage{hyperref}

\usepackage[footnotesize]{caption}
\usepackage[colorinlistoftodos, shadow]{todonotes}

\newcommand{\Rbb}{\mathbb{R}}
\newcommand{\scp}[2]{\langle #1, #2 \rangle}
\newcommand{\Nbb}{\mathbb{N}}
\newcommand{\Zbb}{\mathbb{Z}}
\newcommand{\ZbbMqbin}{\delta\mathbb Z^M}
\newtheorem{theorem}{Theorem}
\newtheorem{definition}{Definition}
\newtheorem{proposition}{Proposition}
\newtheorem{corollary}{Corollary}
\newtheorem{lemma}{Lemma}

\newcommand{\inv}[1]{\frac{1}{#1}}

\newcommand{\tinv}[1]{{\textstyle\frac{1}{#1}}}
\newcommand{\sign}{{\rm sign}\,}

\newcommand{\ud}{\mathrm{d}}

\renewcommand{\leq}{\leqslant}
\renewcommand{\geq}{\geqslant}
\newcommand{\E}{{\mathbb{E}}}

\DeclareMathOperator{\SO}{SO}

\DeclareMathOperator{\iid}{iid}

\DeclareMathOperator{\rad}{rad}
\DeclareMathOperator{\diam}{diam}
\DeclareMathOperator{\buf}{Buffon}

\newcommand{\bb}{\mathbb}

\newcommand{\bs}{\boldsymbol}
\newcommand{\cl}{\mathcal}
\newcommand{\ie}{\emph{i.e.}, }
\newcommand{\eg}{\emph{e.g.}, }

\newcommand{\dN}{\chi_N}
\newcommand{\cN}{\tau_N}
\newcommand{\needle}{{\mathsf N}}
\newcommand{\qbin}{\delta}
\newcommand{\mapf}{\psi}

\newcommand{\rv}{\mbox{r.v.}\xspace}

\title{A Quantized Johnson Lindenstrauss Lemma:\\
The Finding of Buffon's Needle}
\author{Laurent Jacques\thanks{LJ is with the ICTEAM institute, ELEN Department, Universit\'e catholique de Louvain
(UCL), Belgium. Email: \url{laurent.jacques@uclouvain.be}. LJ is funded by Belgian National Science Foundation
(F.R.S.-FNRS).}}

\begin{document}

\maketitle

\begin{abstract}
In 1733, Georges-Louis Leclerc, Comte de Buffon in France, set the
ground of geometric probability theory by defining an enlightening
problem: What is the probability that a needle thrown randomly on a
ground made of equispaced parallel strips lies on two of them? In this
work, we show that the solution to this problem, and its
generalization to $N$ dimensions, allows us to discover a quantized form
of the Johnson-Lindenstrauss (JL) Lemma, \ie one that combines a linear
dimensionality reduction procedure with a uniform quantization of
precision~$\qbin>0$. In particular, given a finite set
$\cl S \subset \Rbb^N$ of $S$ points and a distortion level $\epsilon>0$, as soon as $M > M_0 = O(\epsilon^{-2} \log
S)$, we can (randomly)
construct a mapping from $(\cl S, \ell_2)$ to
$(\ZbbMqbin, \ell_1)$ that approximately preserves the
pairwise distances between the points of~$\cl S$. 
Interestingly, compared to the common JL Lemma, the mapping is quasi-isometric and we
observe both an additive and a multiplicative distortions on the embedded
distances. These two distortions, however, decay as
$O(\sqrt{(\log S)/M})$ when $M$ increases. Moreover, for coarse quantization, \ie for high $\qbin$ compared to
the set radius, the distortion is mainly additive, while for small $\qbin$
we tend to a Lipschitz isometric embedding. 
Finally, we prove the existence of a
``nearly'' quasi-isometric embedding of $(\cl S, \ell_2)$ into $(\ZbbMqbin,
\ell_2)$. This one involves a non-linear distortion of the
$\ell_2$-distance in $\cl S$ that vanishes for distant points in this set. Noticeably, the additive distortion in this
case is slower, and decays as~$O(\sqrt[4]{(\log S)/M})$.
\end{abstract}

\section{Introduction}
\label{sec:introduction}

The Lemma of Johnson-Lindenstrauss (JL) \cite{johnson1984extensions} is a corner stone of (linear)
dimensionality reduction techniques. This result, which can be seen as
a direct consequence of the concentration of measure phenomenon~\cite{dasgupta99elementary}, is at the heart
of many applications in classical search methods for approximate
nearest neighbors \cite{ailon2009fast}, high-dimensional machine learning
\cite{maillard:inria-00419210,fradkin2003experiments}, and compressed
sensing~\cite{JLmeetCS,krahmer2011new}.   

In short, this lemma states that given a
finite set of $S$ points in an $N$-dimensional space, and provided that $M$
scales like $O(\epsilon^{-2} \log S)$ for some allowed distortion
level $\epsilon>0$, there exists a mapping that
projects the elements of this set into a smaller
$M$-dimensional space, without disturbing the pairwise distances
of these points by more than a factor $(1 \pm \epsilon)$.

Mathematically, the classical
formulation of this important lemma is as follows.
\begin{lemma}[Johnson-Lindenstrauss] Given $\epsilon\in (0,1)$,
  for every set $\cl S$ of $S$ points in $\Rbb^N$, if $M$ is such that 
$$
M > M_0 = O(\epsilon^{-2}\log S ),
$$
then there exists a Lipschitz mapping $\bs f: \Rbb^N \to \Rbb^M$ such
that
\begin{equation}
  \label{eq:JL-mapping}
  (1-\epsilon)\,\|\bs u - \bs v\|^2\ \leq\ \|\bs f(\bs u) - \bs f(\bs v)\|^2\ \leq\ (1+\epsilon)\,\|\bs u - \bs v\|^2,
\end{equation}
for all $\bs u, \bs v \in \cl S$. 
\end{lemma}

Beyond this proof of existence, the construction of
(random) Lipschitz mappings from $\Rbb^N$ to $\Rbb^M$ satisfying
\eqref{eq:JL-mapping} is easy \cite{dasgupta99elementary}. In particular, for 
$$
\bs f(\bs u) = \bs \Phi \bs u,
$$  
where $\bs\Phi \in \Rbb^{M\times N}$ is a certain random matrix (\eg
whose independent entries follow identical Gaussian, Bernoulli or sub-Gaussian
distributions \cite{Achlioptas:2003p640}), measure concentration
guarantees that \cite{JLmeetCS}
$$
\bb P\big[ \big| \|\bs \Phi (\bs u - \bs v)\|^2 - \|\bs u - \bs v\|^2 \big| \geq \epsilon
\|\bs u - \bs v\|^2\big]\ \leq\ 2 e^{-M \eta(\epsilon)},
$$ 
where the probability is related to the generation of $\bs \Phi$, and
$\eta$ is a nondecreasing function of $\epsilon \in (0,1)$. For
instance, for $\bs\Phi\sim \cl N^{M\times N}(0,1/M)$, \ie $\bs \Phi \in \Rbb^{M\times N}$ with $\Phi_{ij} \sim_{\rm
iid} \cl N(0,1/M)$, we have $\eta(\epsilon) =
\epsilon^2/2 - \epsilon^3/6 \geq \epsilon^2/3$ \cite{JLmeetCS}. 

Proving the JL Lemma amounts to applying a union bound on all
possible pairs of points $\bs u$ and $\bs v$ taken in $\cl S$. Since
there are no more than ${S \choose 2}\leq S^2/2$ such pairs, the probability that
at least one of them fails to respect \eqref{eq:JL-mapping} is bounded by 
$2 {S \choose 2} e^{-M \eta(\epsilon)} \leq e^{2\log S -M
  \eta(\epsilon)}$. Therefore, as soon as $M > 2\,\eta(\epsilon)^{-1}\log S$,
this probability can be made arbitrarily low. Moreover, 
generating a sequence of $\bs \Phi$ further decreases this probability by hoping
that at least one such matrix respects~\eqref{eq:JL-mapping}; in
the limit, this ensures
the existence of $\bs f$ with probability 1 in Prop.~\ref{prop:quantiz-jl-lemma}.

Combining such linear random mappings with a quantization
procedure $\cl Q$ (\eg uniform or non-uniform) has recently been a
matter of intense research. The implicit
objective of this association is to reduce the amount of
bits required to encode the result of the dimensionality reduction 
\cite{boufounos2011secure}, and to understand the impact of
quantization on the distortion caused by the mapping. 
For instance, the field of 1-bit Compressed Sensing
is interested in reconstructing sparse vectors from the sign of their
random projections \cite{CS1bit,Boufounos2010,Jacques2011,Plan2011,plan2011dimension}. At the heart of this topic lies the extreme
``one-bit'' (or binary) mapping $\bs \mapf_{\rm bin}:\Rbb^N \to \cl
B^M$ with $\cl B=\{\pm 1\}$, \ie
$$
\bs \mapf_{\rm bin}(\bs u) = \sign(\bs \Phi \bs u)
$$ 
for a Gaussian random matrix $\bs \Phi \sim \cl N^{M\times
  N}(0,1)$. Thanks to $\bs \mapf_{\rm bin}$, a set of vectors
of $\Rbb^N$ can be mapped to a
subset of the \emph{Boolean cube}~$\cl
B^M$. For characterizing the distortion introduced by such a mapping,
we must suitably define two new
distances: the normalized Hamming distance $d_H(\bs r,\bs s) = \inv{M}
\sum_i \bb I(r_i \neq s_i)$ between two \emph{binary strings} $\bs r,\bs
s\in\cl B^M$ and the angular distance $d_S(\bs u, \bs v)=\arccos(\|\bs
u\|^{-1}\|\bs v\|^{-1}
\scp{\bs u}{\bs v})$ between two vectors $\bs u,\bs v\in\Rbb^N$. The
use of $d_S$
stems from the vector amplitude loss in the definition of $\bs
\mapf_{\rm bin}$. Within such a context, the following result is known
(its proof is sketched in Sec.~\ref{sec:discussion}).
\begin{proposition}[\cite{GoeWil::1995::Improved-approximation,Jacques2011}]
\label{prop:conc-prop-hamming}
Let $\bs u,\bs v\in \Rbb^N$.  Fix $\epsilon > 0$ and randomly generate
$\bs\Phi\sim\cl N^{M\times N}(0,1)$.  Then we have
\begin{equation}
\label{eq:conc-prop-hamming}
\bb P\left(\ \left|\,d_H\big(\bs \mapf_{\rm bin} (\bs u), \bs \mapf_{\rm bin}(\bs v)\big)\ -\ d_S(\bs u,\bs v)\,\right|\,\leq\,\epsilon\
\right) \geq\ 1 - 2\,e^{-2\epsilon^2 M}, 
\end{equation}
where the probability is with respect to the generation of $\bs\Phi$.
\end{proposition}

Following again a union
bound argument on all pairs of a set $\cl S\subset \Rbb^N$ of size
$S$, for a fixed $\epsilon>0$ and given $M > M_0 =
O(\epsilon^{-2}\log S)$, Prop.~\ref{prop:conc-prop-hamming} induces a
certain embedding of $(\cl S \subset
\Rbb^N, d_S)$ in $(\cl B^M, d_H)$ where, for all $\bs u, \bs v\in \cl S$,
\begin{equation}
\label{eq:conc-prop-hamming-set}
 d_S(\bs x,\bs s) - \epsilon\ \leq\ d_H\big(\bs \mapf_{\rm bin} (\bs
 x),\bs \mapf_{\rm bin} (\bs s)\big)\ \leq\ d_S(\bs x,\bs s) + \epsilon,
\end{equation}
with high probability.

We directly notice two striking differences with the classical
formulation of the JL Lemma: the use of new distance definitions of
course, but more importantly, the presence of an 
error $\epsilon$ that is now additive with respect to the angular distance~$d_S$.
 
Actually, \eqref{eq:conc-prop-hamming-set} shows that 1-bit quantization breaks the \emph{isometric} property of
random linear mappings. These actually become \emph{quasi-isometric} between
the metric spaces $(\cl S, d_S)$ and $(\bs\mapf_{\rm bin}(\cl S)\subset\cl B^M, d_H)$ in the
following sense.
\begin{definition}[\cite{bridsongeometric}]
  \label{def:quasi-isom-def}
  A function $\bs h: \cl X\rightarrow \cl Y$ is called a
  \emph{quasi-isometry} between metric spaces $(\cl X, d_{\cl X})$ and $(\cl Y,
  d_{\cl Y})$ if there exists $C>0$ and $D\geq0$ such that
  $$
  \textstyle \frac{1}{C}\,d_{\cl X}(\bs x, \bs s)\,-\, D\ \leq\ d_{\cl Y}(\bs
  h(\bs x), \bs h(\bs s))\
  \leq\ C d_{\cl X}(\bs x, \bs s) \,+\, D,
  $$ 
  for $\bs x, \bs s \in \cl X$, and $E>0$
  such that $d_{\cl Y}(y,h(\bs x))<E$ for all $y\in \cl Y$.   
\end{definition}
\medskip
 
This paper aims at going beyond the aforementioned binary case. We want to characterize the impact of a more general uniform
quantization $\cl Q$ of
precision (or \emph{bin width})~$\qbin>0$ on a linear dimensionality
reduction procedure. In particular, our objective is to find a mapping $\bs \mapf: \Rbb^N \to
\ZbbMqbin$, combining a
linear random projection from $\Rbb^N$ to $\Rbb^M$ with a certain~$\cl
Q:\Rbb^M\to \ZbbMqbin$, for which the following quasi-isometric
relation is satisfied: 
$$
(1-\epsilon) d(\bs u,\bs v) - \epsilon' \leq d'(\bs \mapf(\bs u), \bs \mapf(\bs v)) \leq (1+\epsilon) d(\bs u,\bs v) + \epsilon',
$$
for all $\bs u,\bs v \in \cl S$, for some distances $d$ and $d'$, and with $\epsilon,\epsilon'>0$
decreasing with $\qbin$ or $M$. This would
generalize nicely the JL Lemma by also showing that, despite is quasi-isometric nature, the mapping is
tighter when the dimensionality $M$ increases, or that it is nearly isometric when $\qbin$ vanishes.  

As it will become clear in Sec.~\ref{sec:quantized-embeddings}, we answer positively to
this quest when $d$ and $d'$ are the $\ell_2$ and $\ell_1$ distances,
respectively, and for $\epsilon \propto \epsilon'$.  Our main result
is as follows. 
\begin{proposition}
\label{prop:quantiz-jl-lemma}
Let $\cl S \subset \Rbb^N$ be a set of $S$ points. Fix $0<\epsilon<1$
and $\qbin >0$. 
For $M > M_0 = O(\epsilon^{-2}\log S)$, there exist a non-linear mapping $\bs \mapf:\Rbb^N\to \ZbbMqbin$ and two
constants $c,c'>0$ such that, for all pairs $\bs u,\bs v\in \cl S$,
\begin{equation}
  \label{eq:quantiz-jl-lemma}
  \textstyle (1 - \epsilon) \|\bs u - \bs
  v\| - c\qbin\epsilon \leq\ \frac{c'}{M} \|\bs \mapf(\bs u) -
  \bs \mapf(\bs v)\|_1
\leq\ 
(1 + \epsilon)\|\bs u - \bs v\| + c \qbin \epsilon.
\end{equation}  
\end{proposition}
More specifically, given a uniform
quantization $\lambda \in \bb R \mapsto \cl Q_\qbin(\lambda) := \qbin \lfloor
\lambda/\qbin \rfloor \in \qbin\bb Z$ (applied componentwise on
vectors), Sec.~\ref{sec:quantized-embeddings} demonstrates that, for
some $C,c'' >0$, if $M
> C \epsilon^{-2} \log S$, then, given a random Gaussian matrix $\bs\Phi \sim \cl N^{M\times N}(0,1)$ and a
uniform random vector (or \emph{dithering} \cite{gray1998quantization,Boufounos2010}) $\bs\xi \sim
\cl U^M([0, \qbin])$ (with $\cl U$ the uniform distribution), the quantized random mapping 
\begin{equation}
  \label{eq:quantiz-map-first}
  \bs x \in \bb R^N\ \mapsto\ \bs\mapf_\qbin(\bs x) := \cl Q_\qbin(\bs\Phi \bs x
  + \bs\xi) \in \ZbbMqbin  
\end{equation}
respects \eqref{eq:quantiz-jl-lemma} with
probability at least $1 - \exp (-c'' \epsilon^2 M)$. 

Prop.~\ref{prop:quantiz-jl-lemma} shows that there exists a quasi-isometric mapping between
$(\cl S \subset \Rbb^N, \ell_2)$ and $(\bs \mapf (\cl S) \subset \ZbbMqbin, \ell_1)$ with
constants $D = c\qbin\epsilon$, $C = 1/(1-\epsilon)\geq 1 + \epsilon$
for $0\leq\epsilon<1$, and finite $E$ in
Def.~\ref{def:quasi-isom-def}. In the rest of this paper, we will
forget these subtleties and say that a relation such as
\eqref{eq:quantiz-jl-lemma} defines a quasi-isometric mapping between 
$(\cl S, \ell_2)$ and $(\ZbbMqbin, \ell_1)$, or equivalently, a
$\ell_2/\ell_1$ quasi-isometric embedding of $\cl S$ in $\ZbbMqbin$.

We clearly see in \eqref{eq:quantiz-jl-lemma} the two expected
distortions: one additive of amplitude $c\qbin\,\epsilon$, and the
other multiplicative and associated to an error factor $(1\pm
\epsilon)$. The additive distortion vanishes if $\qbin$ tends to
zero (whereas the other does not). Moreover, by inverting the relation between $M$
and $\epsilon$, we
observe that both errors decay as $O(\sqrt{\log S/M})$.
In the case of an infinitely fine quantization ($\qbin\to 0$), we
also recover classical embedding results of $(\Rbb^N, \ell_2)$ in
$(\Rbb^M,\ell_1)$ associated to measure concentration in
Banach spaces~\cite{ledoux2005concentration,ledoux1991pbs} (see Sec.~\ref{sec:discussion}).

Notice that Prop.~\ref{prop:quantiz-jl-lemma} generalizes somehow the result
obtained in \cite{boufounos2011secure} for universal binary schemes \cite{Boufounos2010},
\ie when the 1-bit quantizer is \emph{non-regular} and has discontinuous quantization regions. The reason for this is that, despite
its regularity, our quantizer can be seen as a $B$-bit uniform
quantizer where $B$ should be related to $\log_2(\max_{j,\bs u\in\cl S} |(\bs\mapf(\bs u))_j|/\qbin)$. Thus, we show here
that the behavior of the additive distortion of binary quantized
mappings discovered in \cite{boufounos2011secure} is also valid at a
higher number of bits.
\medskip

For reasons that will become clear later, the context that makes
Prop.~\ref{prop:quantiz-jl-lemma} possible\footnote{
  However, as
  mentioned at the end of this Introduction, this is not the only context able to induce Prop.~\ref{prop:quantiz-jl-lemma}.} was already defined in 1733 by Georges-Louis Leclerc, Comte de Buffon in
France. In one of the volumes of his impressive work entitled ``L'Histoire Naturelle\footnote{\url{http://www.buffon.cnrs.fr/?lang=en}}'', this French naturalist stated and solved the following
important problem \cite{buffon1777essai}:
\begin{quote}[English translation of Fig.~\ref{fig:buffon-problem-2d-orig-book} from \cite{hey2010georges}]
``\emph{I suppose that in a room where the floor is simply divided by
  parallel joints one throws a stick (N/A: later called ``needle'')  in the air, and that one of the players bets that the stick will not cross any of the parallels on the floor, and that the other in contrast bets that the stick will cross some of these parallels; one asks for the chances of these two players.}''  
\end{quote}

\begin{figure}[!t]
  \centering
  \subfigure[\label{fig:buffon-problem-2d-orig-book}]{\includegraphics[width=7cm]{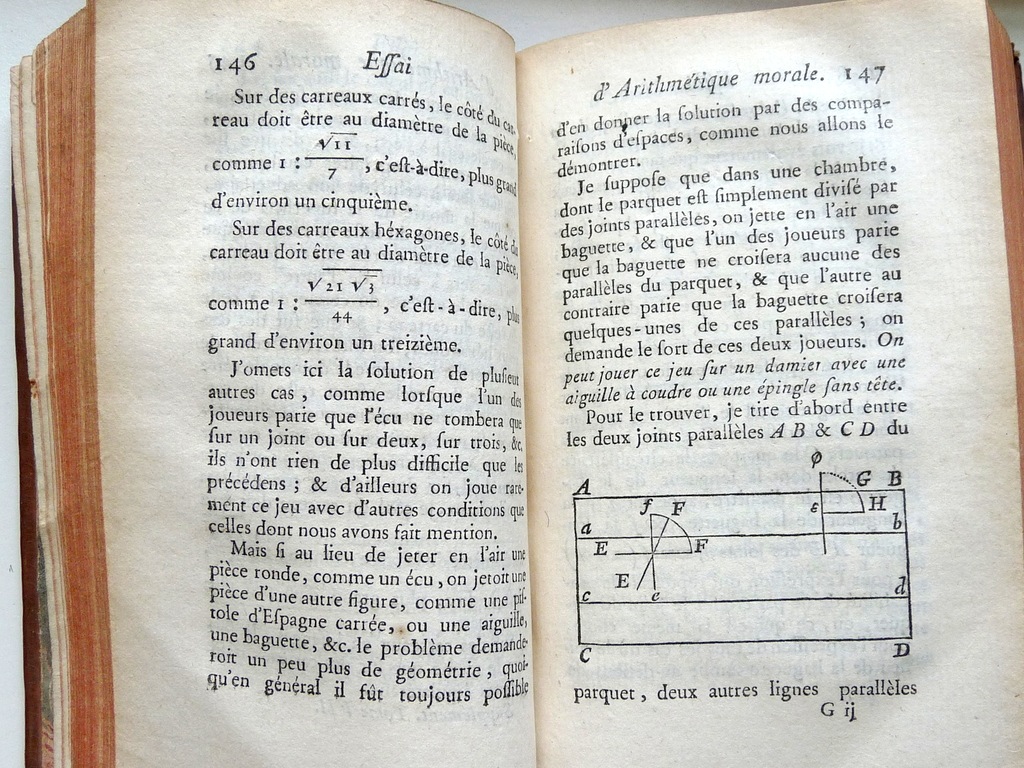}}\hspace{1cm} 
  \subfigure[\label{fig:buffon-problem-2d}]{\includegraphics[width=6cm]{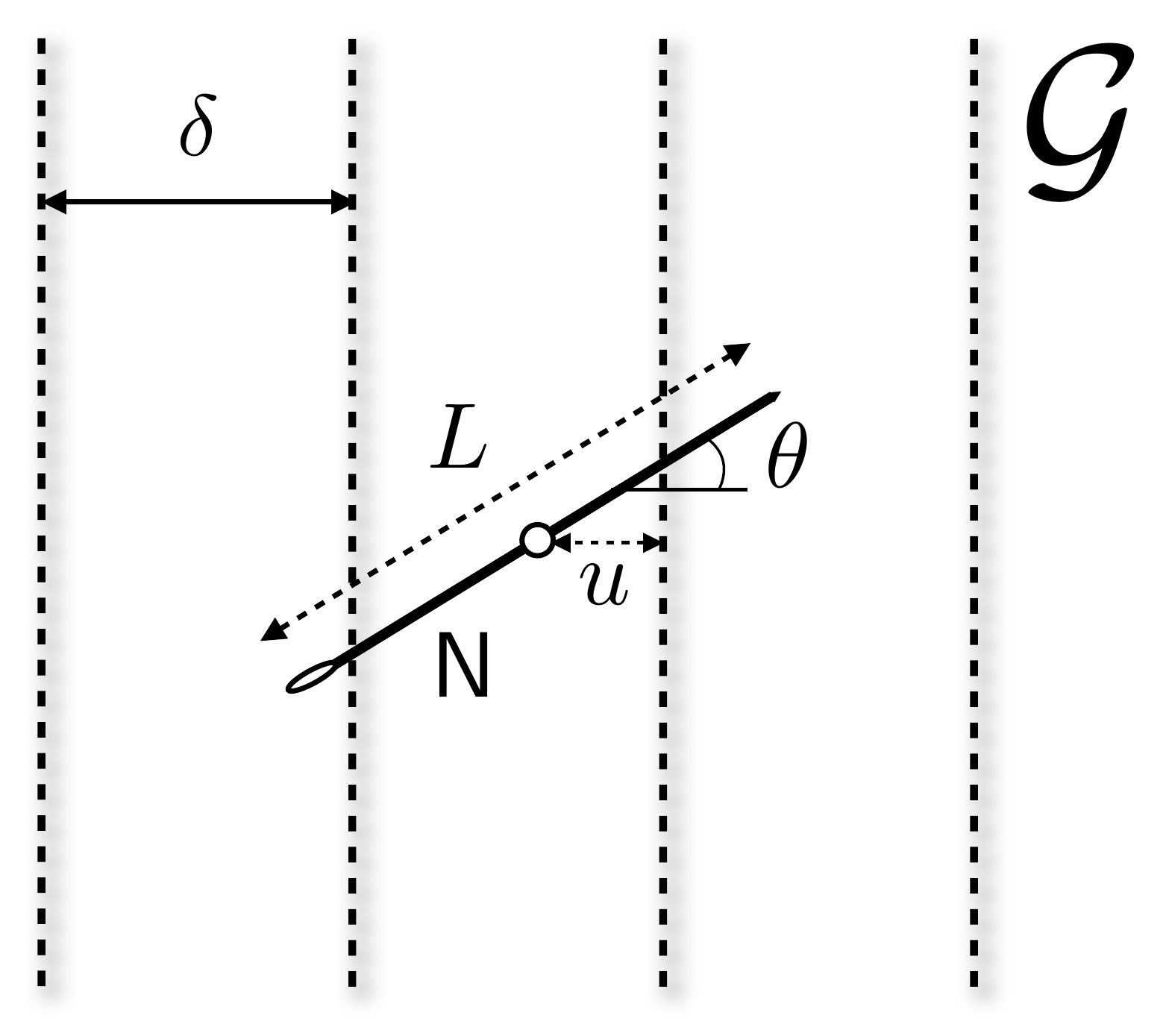}}
  \caption{(a) Picture of \cite[page 147]{buffon1777essai} stating
    the initial formulation of Buffon's needle problem (Courtesy of
    E. Kowalski's blog
    \url{http://blogs.ethz.ch/kowalski/2008/09/25/buffons-needle}). (b)
    Scheme of Buffon's needle problem}
\end{figure}

As explained in Sec.~\ref{sec:init-form-solut}, the solution is astonishingly simple: for a short needle compared to the separation
$\qbin$ between two consecutive parallels (see Fig. \ref{fig:buffon-problem-2d}), the
probability of having one intersection between the needle and the
parallels is equal to the needle length times $\frac{2}{\qbin\pi}$. If
the needle is longer, then this probability is less easy to express but
the expectation of the number of intersections (which can now be
bigger than one) remains equal to this value.   

This problem, and its solution published in 1777
\cite{buffon1777essai}, is considered as the beginning of the
discipline called ``geometrical probability''
\cite{kendall1963geometrical}. Moreover, the solution has also shed new
light on the estimation of $\pi$, \ie by
estimating the probability of intersection on a large number of throws, paving
the way to the well-known stochastic (Monte Carlo) estimation methods.

In this paper, we are going to show that the analysis of Buffon's problem, and its generalization to an $N$-dimensional space, allows us
to specify the conditions surrounding
Prop.~\ref{prop:quantiz-jl-lemma}. As explained in
Sec.~\ref{sec:quantized-embeddings}, the connection between the
existence of a quantized embedding and Buffon's problem is
simple. Forgetting a few technicalities detailed later, \emph{uniformly quantizing the random projections in $\Rbb^M$ of two points in
$\Rbb^N$ and measuring
the difference between their quantized values is fully equivalent to
study the number of intersections made by the segment determined by those two points
(seen as a Buffon's needle) with a parallel grid of
$(N-1)$-dimensional hyperplanes.} 

As an aside to proving Prop.~\ref{prop:quantiz-jl-lemma}, this paper
provides also, to the best of our knowledge, new results on the
behavior of Buffon's needle problem in high-dimensional space. For
instance, we establish a few interesting bounds and asymptotic relations concerning the moments of
the random variable counting the needle/grid intersections (see Sec.~\ref{sec:gener-n-dimens}).\\

In summary, the main contributions of this paper can be considered threefold:
\medskip

\noindent\textbf{\small (C1)} We study the impact of a simple (dithered) quantization on the
Johnson-Lindenstrauss Lemma and show how the introduced distortions
(both additive and multiplicative) decay with~$M$;
\medskip

\noindent\textbf{\small (C2)} We generalize Buffon's needle problem in $N$ dimensions and
  bound all the moments of a discrete distribution ${\rm
    Buffon}(a, N)$ (with $a>0$) counting the
  intersections that a randomly thrown 1-D ``needle'' of length $a$
  makes in $\bb R^N$ with a fixed grid
  of parallel $(N-1)$-hyperplanes spaced by a unit length;
\medskip

\noindent\textbf{\small (C3)} A bridge is built between the characterization of a
  quantized JL Lemma and this generalized Buffon's
  needle problem.
\medskip

Of course, this paper does not claim that (C1) can only be proved thanks to
(C2) and other methods developed on different mathematical tools could
exist. However, we find that the connection (C3) made between (C1) and (C2)
is sufficiently interesting for being presented in this work.

\medskip
The rest of the paper is organized as follows. First, we discuss in
  Sec.~\ref{sec:discussion} our main result, \ie
  Prop.~\ref{prop:quantiz-jl-lemma}, and we identify different distortion regimes of
 our quantized embedding related to the extreme
  values of both $\delta$ and the number of measurements~$M$. Next,
  the initial
problem of Buffon's needle, \ie the core of our developments, and its solution are explained in
Sec.~\ref{sec:init-form-solut} before its $N$-dimensional
generalization developed in Sec.~\ref{sec:gener-n-dimens}.  The
relation between this problem and the existence of an $\ell_2/\ell_1$ quantized
embedding of $\cl S \subset \Rbb^N$ in $\ZbbMqbin$ is then provided in Sec.~\ref{sec:quantized-embeddings}. Finally,
we
provide in Sec.~\ref{sec:extensions} an extension of our analysis that
provides a ``nearly'' quasi-isometric embedding of $(\cl S, \ell_2)$ in $(\ZbbMqbin,
\ell_2)$. This one must be considered with a non-linear distortion of the
$\ell_2$-distance in $\cl S$ that vanishes for large pairwise
distances in this set. Remarkably, the additive distortion in this
mapping decays more slowly with $M$, \ie as~$O((\log S/M)^{1/4})$.

Let us finally acknowledge an anonymous and expert reviewer
  for having pointed out a very elegant and compact proof of
  Prop.~\ref{prop:quantiz-jl-lemma} that does not rely on our
  generalization of Buffon's needle problem. In short, this proof actually uses the
  properties of sub-Gaussian random variables
  \cite{IntroNonAsRandom} in order to \emph{(i)} characterize the sub-Gaussian
  nature of $\bs\mapf_\qbin(\bs u) - \bs\mapf_\qbin(\bs v)$ for any
  pair of vectors $\bs u,\bs v \in \bb R^N$, \emph{(ii)} demonstrate the concentration
  properties of the $\ell_1$-norm of this difference around its mean, and
  \emph{(iii)} showing that this mean can be characterized by the
  two-dimensional case of Buffon's needle problem. This proof is
  reported in Appendix~\ref{sec:altern-proof-prop2}. In order to
  make it self-contained, we have also briefly recalled there the
  definition and main properties of sub-Gaussian
distributions. We believe that the tools developed in this alternative
proof provide a powerful analysis of quantized random projections that
can be useful for the interested readers.

\paragraph*{Conventions:} Most domain dimensions are denoted
by capital roman letters, \eg $M, N, \ldots$ Vectors, vector functions and matrices are associated to bold
symbols, \eg $\bs \Phi \in \Rbb^{M\times N}$ or $\bs u \in \Rbb^M$, while lowercase light letters are associated to scalar
values. The $i^{\rm th}$
component of a vector (or a vector function) $\bs u$ reads either $u_i$ or $(\bs u)_i$,
while the notation $\bs u_i$ refers to the $i^{\rm
  th}$ element of a set of vectors. The set of indices in
$\Rbb^D$ is $[D]=\{1,\,\cdots,D\}$. The scalar product between two vectors
$\bs u,\bs v \in \Rbb^{D}$ for some dimension $D\in\Nbb$ is denoted
equivalently by $\bs u^T \bs v = {\bs u \cdot \bs v}=\scp{\bs u}{\bs
  v}$. For any $p\geq 1$, the $\ell_p$-norm of $\bs u$ is $\|\bs u\|_p^p = \sum_i |u_i|^p$ with
$\|\!\cdot\!\|=\|\!\cdot\!\|_2$.  We will abuse the notation ``$\ell_p$''
to either denote the $\ell_p$-norm as above or the $\ell_p$-distance
(or metric) between two points $\bs u, \bs v\in\Rbb^N$ defined by
$\|\bs u - \bs v\|_p$ (\eg for defining a metric space $(\cl X \subset
\Rbb^D, \ell_p)$). The event indicator function $\bb I$
is defined as $\bb I(A) = 1$ if $A$ is verified and 0 otherwise.  A
uniform distribution over $\cl I \subset \Rbb$ is denoted by $\cl U(\cl I)$. A
random matrix $\bs \Phi \sim \cl D^{M\times N}(\Theta)$ is an $M\times N$
matrix with entries distributed as $\Phi_{ij} \sim_{\iid} \cl
D(\Theta)$ given the distribution parameters $\Theta$ of $\cl D$ (\eg $\cl
N^{M\times N}(0,1)$ or $\cl U^{M\times N}([0,1])$). A
random vector in $\Rbb^M$ following $\cl D(\Theta)$ is defined by $\bs v \sim \cl D^{M}(\Theta)$.
Given two random variables $X$ and $Y$, the notation $X \sim
Y$ means that $X$ and $Y$ have the same distribution.
The probability of an event $\cl E$ is denoted $\bb P(\cl E)$.  The
diameter of a finite set $\cl S\subset \Rbb^N$ of cardinality $|\cl S|$ is $\diam \cl S = \max_{\bs u,\bs v\in\cl S}
\|\bs u -\bs v\|$ and its radius is $\rad \cl S = \max_{\bs u\in\cl S}
\|\bs u\|$. The $(N-1)$-sphere in $\Rbb^N$ is $\bb S^{N-1}=\{\bs x\in\Rbb^N:
\|\bs x\|=1\}$.  For asymptotic relations, we use the common Landau
family of notations, \ie the symbols $O$, $\Omega$ and $\Theta$ (their exact definition can be found in
\cite{knuth1976big}). The positive thresholding function is defined by
$(\lambda)_+ := \tinv{2}(\lambda + |\lambda|)$ for
any~$\lambda\in\Rbb$.

\section{Discussion}
\label{sec:discussion}

How can we analyze the distortions induced by the quasi-isometric
mapping provided by Prop.~\ref{prop:quantiz-jl-lemma}? 
Interestingly, we can identify three key regimes, depending on
the values of $\delta$ and $M$, where a quantized embedding respecting
\eqref{eq:quantiz-jl-lemma} displays different typical behaviors.
\paragraph*{(a) Nearly isometric regime} Under a fine quantization scheme, \ie if 
\begin{equation}
  \label{eq:nu-S-def}
\qbin\ \ll\ \nu_{\cl S} := \min_{\bs u, \bs v\,\in\,\cl S:\ \bs u\neq \bs
  v} \|\bs u - \bs v\|,  
\end{equation}
Eq. \eqref{eq:quantiz-jl-lemma} essentially provides a Lipschitz embedding
of $(\cl S, \ell_2)$ in $(\bs\mapf(\cl S), \ell_1)$. This is sensible since, considering the mapping $\bs\mapf_\qbin(\bs x) := \cl Q_\qbin(\bs\Phi \bs x
  + \bs\xi)$ defined in \eqref{eq:quantiz-map-first}, for
such a fine quantization, the corresponding distortion almost
disappears, \ie $\cl Q(\lambda) \simeq \lambda$ for any $\lambda \gg \qbin$,
and it is known that, for a Gaussian random matrix $\bs \Phi \sim \cl
N^{M\times N}(0,1)$ and two fixed vectors $\bs u$ and $\bs v$,
\begin{equation}
  \label{eq:lipschitz-l2-l1-mapping}
  (1 - \epsilon)\,\|\bs u - \bs
  v\|\ \leq\ \tfrac{\sqrt{\pi}}{\sqrt 2\, M}\|\bs \Phi \bs u - \bs
  \Phi \bs v\|_1\ \leq\ 
  (1 + \epsilon)\|\bs u - \bs v\|,  
\end{equation}
with probability higher than $1 - 2e^{-\inv{2}\epsilon^2 M}$.

Indeed, as explained for instance in \cite[Appendix A]{Jacques2010}, this is a
simple consequence of the following result due to Ledoux and Talagrand.
\begin{proposition}[Ledoux, Talagrand \cite{ledoux1991pbs} (Eq. 1.6)] 
If $F$ is
Lipschitz with constant\footnote{The Lipschitz constant of $F$ is defined as $\|F\|_{\rm Lip}
\triangleq \sup_{\bs x, \bs y\,\in\,\Rbb^M,\ \bs x \neq \bs
  y}\,{|F(\bs x) -
  F(\bs y)|}\,/\,{\|\bs x-\bs y\|_2}$.} $\lambda=\|F\|_{\rm Lip}$, then, for a random vector
$\bs \zeta\in\Rbb^M$ with $\zeta_i\sim_{\rm iid} \mathcal{N}(0,1)$ (\ie
$\bs \zeta \sim \cl N^{M}(0,1)$), 
\begin{equation*}
\bb P\big[\,| F(\bs \zeta) - \mu_F | > r\,\big]\ \leq\ 2
e^{-\inv{2} r^2 \lambda^{-2}},\quad {\rm for}\ r>0,
\end{equation*}
with $\mu_F = \E[F(\bs \zeta)]$.
\end{proposition}
For a Gaussian random matrix $\bs \Phi \sim \cl N^{M\times N}(0,1)$, the vector
$\bs \zeta = \|\bs u -\bs v\|^{-1} \bs \Phi(\bs u -\bs v)$ is
distributed as $\cl N^{M}(0,1)$. Taking $F(\cdot) = \|\cdot\|_1$
with $\|F\|_{\rm Lip} = \sqrt M$ and $r=M\epsilon$ with $\epsilon> 0$
provides \eqref{eq:lipschitz-l2-l1-mapping} since $\mu_F = M
\sqrt{\tfrac{2}{\pi}}$. 

As explained before, the result \eqref{eq:lipschitz-l2-l1-mapping} is
easily extendable (from a union bound argument) to
the embedding of a finite set $\cl S$ of $S$ points in $\Rbb^N$ provided $M > M_0 =
O(\epsilon^2 \log S)$. Noticeably, Prop.~\ref{prop:quantiz-jl-lemma}
converges exactly to this isometric
mapping if $\qbin \ll \nu_{\cl S}$.

\paragraph*{(b) Quasi-isometric binary regime} In the case where
$\qbin$ is greater than the diameter $\diam\cl S$, \ie the greatest
distance between any pair of points in this set, then the quantization
distortion dominates and the quantized embedding reduces to a
quasi-isometric embedding. Indeed, for such a situation, we 
reach 
\begin{equation}
\textstyle \|\bs u - \bs
  v\| - (1+ c)\qbin \epsilon\ \leq\ \frac{c'}{M}\|\bs \mapf(\bs u) -
  \bs \mapf (\bs v)\|_1\ \leq\ \|\bs u - \bs v\| + (1+c)\qbin\epsilon, 
\end{equation}
since then $\|\bs u - \bs v\| \leq \qbin$. This is reminiscent of the
observations made in
\cite{GoeWil::1995::Improved-approximation,Jacques2011} about the
embedding properties of ``binarized'' random projections. As explained in the Introduction, given
$\bs u,\bs v\in \Rbb^N$ and $\epsilon > 0$, if we randomly generate
$\bs\Phi$ as $\cl N^{M\times N}(0,1)$, then, from \eqref{eq:conc-prop-hamming},
$$
d_S(\bs u,\bs v)\,-\,\epsilon\ \leq\ d_H\big(\sign(\bs\Phi\bs
u),\sign(\bs\Phi\bs v)\big)\ \leq\ d_S(\bs u,\bs v)\,+\,\epsilon,
$$
with probability exceeding $1 - 2\,e^{-2\epsilon^2 M}$. In short,
this result amounts to first showing that the signs of $\bs \varphi_j^T\bs u$
and $\bs \varphi_j^T\bs v$ differ with a probability equal to $d_S(\bs u,\bs
v)$ for any $j\in[M]$, and second to observing that the sum of all
such signs
collected at every $j$ (as performed in the Hamming
distance $d_H$) behaves as a Binomial random variable of $M$ trials and
probability $d_S(\bs u,\bs
v)$. This kind of random variable is known to concentrate quickly around its mean $d_S(\bs u,\bs
v)$ from a simple application of the Chernoff-Hoeffding inequality \cite{Jacques2011}.

In this considered case, the 1-bit quantization of the random
projections performed by the sign operator is not strictly equivalent to our quantization scheme
defined in \eqref{eq:quantiz-mappin-def-with-dithering} in that there is
no \emph{dithering}. This absence imposes the definition of other
distances $d_S$ and $d_H$, while in our case the dither allows one to
recover an Euclidean ($\ell_2$) distance in $\Rbb^N$ rather than the angular
one. However, for both kinds of quantizations, we do observe the same
quasi-isometric behavior with a dominant additive distortion $\epsilon$.

\paragraph*{(c) High measurement regime} This is possibly the most
interesting regime since it displays some ``blessing of dimensionality''
for tightening the two quasi-isometric distortions as $M$ increases. It was formerly observed in \cite[p.~3]{BR_DCC13} that for
a scalar uniform quantizer $\cl Q$ such as ours, if $M > M_0 = 
O(\epsilon^{2} \log M)$, the JL Lemma induces \emph{a priori} a
quasi-isometric mapping with a much looser additive
distortion. Indeed, given $\bs \Phi \sim \cl N^{M\times N}(0,1)$ (with
this prescribed $M$), for
any points $\bs u,\bs v\in\cl S$ we have 
\begin{equation}
\textstyle (1 - \epsilon) \|\bs u - \bs
  v\| -c\qbin\ \leq\ \tinv{\sqrt{M}}\|\cl Q(\bs \Phi \bs u) - \cl Q(\bs \Phi \bs v)\|\ \leq\ (1 + \epsilon)\|\bs u - \bs v\| + c\qbin,  
\end{equation}
for some $c>0$.  For our uniform quantizer $\cl Q$ and any dithering
$\bs \xi \in \Rbb^M$, this is easily obtained from the relation $|\cl Q(\lambda) - \lambda|
\leq \lambda/2$ and from 
\begin{equation}
\textstyle \|\bs \Phi
\bs u - \bs \Phi \bs v\|^2 - \inv{2}\,M \qbin^2\ \leq\ \|\cl Q(\bs \Phi
\bs u + \bs \xi) - \cl Q(\bs \Phi \bs v + \bs \xi)\|^2\ \leq\ \|\bs \Phi
\bs u - \bs \Phi \bs v\|^2 + \inv{2}\,M \qbin^2,  
\end{equation}
with $\|\bs \Phi \bs u - \bs \Phi \bs v\|$
close to $\|\bs u -\bs v\|$ up to a distortion factors $(1\pm
\epsilon)$ by the JL Lemma.  Notice that taking the square root of this inequality for lowering the
power 2 is not a problem since $(a - b) \leq (a^2 - b^2)^{1/2}$ if $a>b>0$ and $(a^2+b^2)^{1/2} <
a+b$ for any $a,b>0$.

Similarly, introducing the quantization in the $\ell_2/\ell_1$
isometric
embedding explained in \eqref{eq:lipschitz-l2-l1-mapping} has also the same
impact since 
\begin{equation}
\textstyle \|\bs \Phi
\bs u - \bs \Phi \bs v\|_1 - \inv{2}\,M \qbin\ \leq\ \|\cl Q(\bs \Phi
\bs u + \bs \xi) - \cl Q(\bs \Phi \bs v + \bs \xi)\|_1\ \leq\ \|\bs \Phi \bs u - \bs \Phi \bs v\|_1 + \inv{2}\,M \qbin. 
\end{equation}

In both situations, the additive error induced by the quantization is constant with $M$.
As expressed by
Prop.~\ref{prop:quantiz-jl-lemma} (and later in
Prop.~\ref{prop:final-quant-embed-prop}), our analysis shows that there
exists a mapping for which the
same error actually scales as $O(\qbin/\sqrt M)$, \ie the finding of
Buffon's needle helped us to reduce that distortion by a factor $\sqrt
M$.   

\section{Buffon's needle problem}
\label{sec:buff-needle-probl}

\subsection{Initial formulation and solution}
\label{sec:init-form-solut}

Let us rephrase Buffon's needle problem stated in the Introduction in a more formal way. Let $\cl G\subset \Rbb^2$ be
a set of equispaced parallel lines in $\Rbb^2$, two consecutive lines being
separated by a distance $\qbin>0$. Let a needle $\needle$ of length
$L$ be thrown
uniformly at random on the plane $\Rbb^2$: its orientation $\theta$ is drawn uniformly at random on the circle~$[0,
2\pi]$, while from the $\qbin$-periodicity of
$\cl G$, the distance $u$ of the needle's midpoint to the closest
line is a uniform random variable over $[0, \qbin/2]$ (see Fig.~\ref{fig:buffon-problem-2d}). 

Buffon's needle problem then amounts to computing the probability $P$ that
$\needle(u,\theta)\cap \cl G\neq \emptyset$. As a matter of fact, this
probability is easily estimated since, conditionally to the knowledge
of $\theta$, there is at least one 
intersection if $2u\leq L|\cos\theta|$. Therefore, we find
\begin{align}
P&\textstyle = 
\tfrac{1}{\pi\qbin}\,\int_0^{2\pi} \ud \theta
\int_0^{\qbin/2} \bb I\big(2\min(u,\qbin -u)\leq L|\cos\theta|\big)\
\ud u\nonumber\\
&\label{eq:proba-1inter-2d}\textstyle = \tfrac{4}{\pi\qbin}\,\int_0^{\pi/2} \ud \theta
\int_0^{\inv{2}\min(\qbin,\,L\cos\theta)} \ud u.
\end{align}
We observe that if $L<\qbin$, then $L\cos\theta < \qbin$ and
$P=\tfrac{2L}{\pi\qbin}$, while if $L\geq \qbin$, the solution reads
$
P = \tfrac{2}{\pi} \theta_1 + \tfrac{2L}{\pi\qbin}(1-\sin\theta_1)
$
with $\cos\theta_1 = \tfrac{\qbin}{L}$.

Notice that if $L<\qbin$, only one intersection is possible, and if
$X$ denotes the random variable associated to the occurrence of such
an intersection, we have therefore $\bb E X = P = \frac{2L}{\pi \qbin}$. 
Interestingly, for any $L> 0$, this expectation still keeps
the same value. 
\begin{proposition}[\cite{ramaley1969buffon}]
\label{prop:expect-intersect}
Let $X$ be the
discrete random variable counting the number of intersections of $\needle$ with~$\cl G$, \ie $X = |\{\needle(u,\theta) \cap \cl G\}|$ where $u$
and $\theta$ are two random variables defined as above. Then, writing $a=L/\qbin$,
$$
0 \leq X \leq \lfloor a \rfloor + 1\quad \text{and}\quad  \bb E X =
\tfrac{2}{\pi}a.
$$   
\end{proposition}
\begin{proof}
We follow the spirit of the proof given in \cite{ramaley1969buffon}. The domain of $X$ is obvious from the problem definition. For
estimating the expectation, let us observe that the needle $\needle$ can
always be considered as being made of two joint needles $\needle_1$ and
$\needle_2$ of lengths $L_1$ and $L_2$ ($L_1+L_2=L$). If $X_1$ and $X_2$
are the random variables counting their respective intersection with $\cl G$, we
have $X=X_1 + X_2$. Therefore, since $\bb E X$ necessarily depends on
$L$ through some nondecreasing function $h$, we find $h(L)=\bb E X = \bb E(X_1 +
X_2) = \bb E X_1 + \bb E X_2 = h(L_1) + h(L_2)$. This shows that
$h(L)=c L$ for some $c>0$ independent of $L$. From the knowledge of
$\bb E X$ for $L<\qbin$, we deduce that $c = \frac{2}{\pi\qbin}$.
\end{proof}

Surprisingly enough, this proposition still holds if the needle is replaced
by any smooth curve of length
$L$ \cite{ramaley1969buffon}. Indeed, such curve can always be approximated by a piecewise
linear contour with arbitrary small error and the proof above does not depends on a possible
bending of the $\needle_1$ and $\needle_2$. However, the distribution
of $X$ does depend on the curve shape. 

Let us specify now what is known of the distribution of $X$ in the
case of a (straight) needle.
\begin{proposition}[{\cite[pp.\ 72--73]{kendall1963geometrical}\cite{diaconis1976buffon}}] 
\label{prop:distrib-X}
Given $a=L/\qbin$, define the angles $\theta_k\in [0,\pi/2]$ such that $\cos \theta_k =
k/a$ for $0\leq k \leq n$ with $n=\lfloor a\rfloor$,
$\cos \theta_{k}=0$ for $k<0$ and $\cos \theta_{k}=1$ for all $k>n$.
The distribution of $X \in [n+1]$ is determined by the probabilities 
\begin{align}
p_k = \bb P(X=k)&=\kappa_{k+1}+\kappa_{k-1}-2\kappa_k,
\end{align}
with $\kappa_k=(2a\sin\theta_k/\pi) - (2k\theta_k/\pi)$. 
\end{proposition}
\begin{proof}
This proof only differs in notations from the one of
\cite[pp. 72-73]{kendall1963geometrical}. Let $n=0$,
$p_0=1-P$ with $P$ computed in \eqref{eq:proba-1inter-2d}. For $\theta$ fixed, the conditional probability of having $n+1$ intersections reads
$a|\cos\theta| - n$ if $\theta\leq \theta_n$, and 0
otherwise. Therefore, we have
$$
p_{n+1} = \tfrac{2}{\pi}\,\int_0^{\theta_n} (a\cos\theta - n)\,\ud
\theta = \tfrac{2a}{\pi}\sin\theta_n - \tfrac{2n}{\pi}\theta_n.\\
$$
For $1\leq k\leq n$, there are $k$ intersections if $\theta_{k+1} \leq
\theta \leq \theta_{k-1}$. Thus, the conditional probability reads $(k+1 -
a\cos\theta)$ if $\theta_{k+1} \leq
\theta \leq \theta_{k}$, and $(a\cos\theta - k+1)$ if $\theta_{k} \leq
\theta \leq \theta_{k-1}$. Therefore, 
\begin{align*}
p_{k}&\textstyle = \tfrac{2}{\pi}\,\int_{\theta_{k+1}}^{\theta_k} (k+1 - a\cos\theta)\,\ud
\theta\\
&\textstyle \qquad\qquad + \tfrac{2}{\pi}\,\int_{\theta_{k}}^{\theta_{k-1}}
(a\cos\theta - k+1)\,\ud
\theta\\
&= \tfrac{2a}{\pi}(\sin\theta_{k+1} + \sin\theta_{k-1} -
  2\sin\theta_{k})\\
&\qquad\qquad- \tfrac{2}{\pi}((k+1)\theta_{k+1} + (k-1)\theta_{k-1}
-2k\theta_{k}).
\end{align*}
The rest of the proof consists in expressing these results in terms of $\kappa_k$.
\end{proof}

An analysis of the other properties of the random variable~$X$ (\eg
characterizing its moments) is postponed after the discussion of the
multidimensional generalization of Buffon's needle problem.

\subsection{$N$-dimensional generalization}
\label{sec:gener-n-dimens}

How does Buffon's needle problem generalize in an $N$-dimensional
space? More precisely, what phenomena do we observe on the ``random
throw''\footnote{Assuming of course that we can throw an object in an $N$-dimensional
  space so that it stops in a fixed position of $\Rbb^N$, as it stops on the floor of the
  2-dimensional formulation.} of a 1-dimensional needle $\needle$ of length $L$ on an infinite set $\cl G$
of equispaced parallel hyperplanes of dimension $N-1$ separated by a
distance $\qbin>0$?

In $N$ dimensions, the position of the needle relatively to $\cl G$
can again be determined by its distance $u \in [0,\qbin/2]$ to the closest hyperplane of
$\cl G$, while its orientation can be characterized by a set of $(N-1)$ angles
$\{\theta,\phi_{1},\phi_2,\cdots,\phi_{N-3}\}$ on $\bb S^{N-1}$. These include the angle $\theta\in[0,\pi]$
measured between the needle and the normal vector orthogonal to all
hyperplanes\footnote{Notice that, conversely to the two-dimensional
  analysis, this angle $\theta$ covers now the half
  circle $[0,\pi]$, the other angles guaranteeing that all
  orientations in $\bb S^{N-1}$ can be obtained.}, while
the others range as $\phi_k\in [0,\pi]$ for $1\leq k \leq N-3$ and
$\phi_{N-2}\in[0,2\pi]$. We recall that in this \emph{hyperspherical} system of
coordinates, the $(N-1)$-sphere $\bb S^{N-1}$
is measured by 
\begin{equation}
\textstyle \sigma(\bb S^{N-1}) = \int_0^\pi (\sin\theta)^{N-2} \ud \theta\ \big(\int_0^\pi
(\sin\phi_1)^{N-3} \ud \phi_1 \ \cdots\ \int_0^\pi
\sin\phi_{N-3}\, \ud \phi_{N-3} \int_0^{2\pi}
\ud \phi_{N-2} \big),  
\end{equation}
where $\sigma(\cdot)$ denotes the rotationally invariant area measure
on the $(N-1)$-sphere. 
  
The first question we can ask ourselves is how the expectation of 
$X=|\needle \cap \cl G|$ evolves in this multidimensional setting. Following the same argumentation of the
previous section, we must still have $\bb E X \propto a$, but what is
now the
proportionality factor? 
\begin{proposition}
In the $N$-dimensional Buffon's needle problem, the expected number of
intersections between the needle and the hyperplanes reads
\begin{equation}
  \label{eq:buffon-expect}
\textstyle\bb E X = \cN a,\quad \text{with}\quad \cN = \frac{\Gamma(\frac{N}{2})}{\sqrt{\pi}\,\Gamma(\frac{N+1}{2})},
\end{equation}
$\tau_2=\frac{2}{\pi}$ and $\tau_3=\frac{1}{2}$.
\end{proposition}
\begin{proof}
As for the proof of Prop.~\ref{prop:expect-intersect}, determining $\cN$
can be done for $L<\qbin$ where $\bb E X$ matches the
probability of having one intersection. In this case, following the
determination of this probability for the two-dimensional case (Sec.~\ref{sec:init-form-solut}), we can say that, conditionally to the knowledge
of $\theta$ and of the $N-2$ other angles $\{\phi_1,\cdots,\phi_{N-2}\}$, there is an
intersection if $2u\leq L|\cos\theta|$. Therefore, defining $I_k:=\int_0^\pi
(\sin\alpha)^k\ud\alpha$ with $\sigma(\bb S^{N-1})=2I_0\cdots I_{N-2}$ and
considering the periodicity of $|\!\cos \theta|$, the
probability $P_N$ of having one intersection generalizes as
\begin{align*}
P_N&\textstyle= \tfrac{2}{\qbin\sigma(\bb S^{N-1})}\,\int_0^{\pi} (\sin\theta)^{N-2}\,\ud
\theta\ \big(\int_0^\pi
(\sin\phi_1)^{N-3}\,\ud\phi_1\,\cdots\\
&\textstyle\qquad\qquad\,\cdots \int_0^\pi
\sin\phi_{N-3}\,\ud\phi_{N-3}\,\int_0^{2\pi}\,\ud\phi_{N-2}\big)\\
&\textstyle\qquad\qquad\qquad\times\ \int_0^{\qbin/2} \bb I\big(2u\leq L |\!\cos\theta|\big)\
\ud u\\
&\textstyle= \tfrac{4}{\qbin I_{N-2}}\,\int_0^{\pi/2} (\sin\theta)^{N-2}\,\ud
\theta\ 
\int_0^{\qbin/2} \bb I\big(2u\leq L \cos\theta\big)\
\ud u\\
&\textstyle= \tfrac{2a}{I_{N-2}}\,\int_0^{\pi/2} (\sin\theta)^{N-2}\,\cos \theta\,\ud
\theta\ =\ \tfrac{2a}{(N-1)I_{N-2}}.
\end{align*}
Since $I_k=\sqrt \pi\,\Gamma(\frac{k+1}{2})/\Gamma(\frac{k}{2} + 1)$ and $\bb E X =
P_N$ for $a < 1$, we find
$$
\cN\ =\ \tfrac{2}{(N-1)I_{N-2}}\ =\
\tfrac{2\Gamma(\frac{N}{2})}{\sqrt{\pi}\,(N-1)\,\Gamma(\frac{N-1}{2})}
= \tfrac{\Gamma(\frac{N}{2})}{\sqrt{\pi}\,\Gamma(\frac{N+1}{2})}.
$$
The values for $\tau_2$ and $\tau_3$ come from the evaluations
$\Gamma(1)=1$, $\Gamma(1/2)=\sqrt{\pi}$ and $\Gamma(3/2)=\sqrt{\pi}/2$.
\end{proof}
To the best of our knowledge, $\cN$ was only known for the case $N=2$ and $N=3$
(see \cite[pp. 70 and 77, respectively]{kendall1963geometrical}). This
quantity behaves as follows.
\begin{proposition}
\label{prop:bounds-EX}
In the $N$-dimensional Buffon's needle problem, 
$$  
\textstyle\ \sqrt{\tfrac{2}{\pi}} \, (N+1)^{-\frac{1}{2}}
\ \leq\ \cN\ \leq\ \sqrt{\tfrac{2}{\pi}} \,  (N-1)^{-\frac{1}{2}}
$$
so that $\bb E X = \Theta(a / \sqrt{N})$.  
\end{proposition}
\begin{proof}
Since $\cN = \tinv{a}\,\bb E X$, this is a direct consequence of the inequality 
$(\frac{2N-3}{4})^{1/2}$ $\leq
{\Gamma(\frac{N}{2})}/{\Gamma(\frac{N-1}{2})}$ $\leq
(\frac{N-1}{2})^{1/2}$ and of the fact that ${(N-\frac{3}{2})^{\frac{1}{2}}}/{(N-1)} \geq 1/{\sqrt{N+1}}$
for $N \geq 2$.  
\end{proof}

We find useful to introduce right now the following general quantity which takes the value
$\cN$ as a special case:
\begin{equation}
  \label{eq:dN-def}
  \dN(x)\ :=\ \frac{\Gamma(x+\inv{2}) \Gamma(\frac{N}{2})}{\sqrt \pi\,
    \Gamma(\frac{N}{2}+ x)}. 
\end{equation}
We can compute $\dN(0)=1$, $\dN(\tinv{2}) = \cN$ and $\dN(1) =
\inv{N}$. The importance of $\dN$, and the notation simplification
brought by its introduction, will become clear later.

Having established how the expectation of $X$ behaves, can we go
further and characterize its distribution as in the two-dimensional case? A positive answer is given in the following proposition. 

\begin{proposition} 
Given $a=L/\qbin$ and the angles $\theta_k\in [0,\pi/2]$ defined in
Prop.~\ref{prop:distrib-X}.
The distribution of $X \in [n+1]$ is determined by the probabilities 
\begin{align}
p_k = \bb P(X=k)&=\kappa_{k+1}+\kappa_{k-1}-2\kappa_k,
\end{align}
with $\kappa_k= \cN\,a\,(\sin\theta_k)^{N-1} -
k\,\cN\,J_N(\theta_k)$ and $J_N(\alpha) :=
(N-1)\int_0^\alpha (\sin \theta)^{N-2}\,\ud\theta$. \\[2mm]
We denote the (discrete) distribution determined by such probabilities as $\buf(a,N)$. 
\end{proposition}
\begin{proof}
The proof consists in considering the hyperspherical coordinates
defined in the demonstration of Prop.~\ref{prop:distrib-X}. For
$k=0$, we must only estimate 
$p_0=1-P_N$ for any value of $a$. For $a<1$, we know that $P_N=\cN a$, while for
$a>1$,  
\begin{align*}
P_N&\textstyle= \tfrac{4}{I_{N-2}\qbin}\,\int_0^{\pi/2} (\sin\theta)^{N-2}\,\ud
\theta\ 
\int_0^{\frac{\delta}{2}} \bb I\big(2u\leq L \cos\theta\big)\
\ud u\\
&\textstyle=\tfrac{4}{I_{N-2}\qbin}\,\big(\tfrac{\qbin}{2} \int_0^{\theta_1} (\sin\theta)^{N-2}\,\ud
\theta\\ 
&\textstyle\qquad\qquad\qquad+ \tfrac{L}{2}\,\int_{\theta_1}^{\pi/2} (\sin\theta)^{N-2}\,\,\cos\theta\,\ud
\theta\big)\\
&\textstyle=\cN\,J_N(\theta_1) 
+ \cN\,a\, (1 - (\sin\theta_1)^{N-1}).
\end{align*}
For $k=n+1$, considering $\theta$ fixed, the conditional probability of having $n+1$ intersections reads
$a|\cos\theta| - n$ if $\theta\leq \theta_n$ and 0 otherwise. Therefore, 
\begin{equation}
\textstyle p_{n+1} = \tfrac{2}{I_{N-2}}\,\int_0^{\theta_n} (a\cos\theta - n)\,(\sin\theta)^{N-2}\,\ud
\theta\ =\ \cN\,a\,(\sin\theta_n)^{N-1} - \cN\,n\,J_N(\theta_n).  
\end{equation}
For $1\leq k\leq n$, there are $k$ intersections if $\theta_{k+1} \leq
\theta \leq \theta_{k-1}$. The conditional probability reads $(k+1 -
a\cos\theta)$ if $\theta_{k+1} \leq
\theta \leq \theta_{k}$, and $(a\cos\theta - k+1)$ if $\theta_{k} \leq
\theta \leq \theta_{k-1}$. Therefore, 
\begin{align*}
p_{k}&\textstyle= \tfrac{2}{I_{N-2}}\,\int_{\theta_{k+1}}^{\theta_k} (k+1 - a\cos\theta)\,(\sin\theta)^{N-2}\ud
\theta\\
&\textstyle\qquad\qquad  + \tfrac{2}{I_{N-2}}\,\int_{\theta_{k}}^{\theta_{k-1}}
(a\cos\theta - k+1)\,(\sin\theta)^{N-2}\ud
\theta\\
&\textstyle= \cN\, a\,\big((\sin\theta_{k+1})^{N-1} + (\sin\theta_{k-1})^{N-1} -
2(\sin\theta_{k})^{N-1}\big)\\
&\quad \textstyle- \cN\big (\ (k+1) J_N(\theta_{k+1})\ + \\
&\textstyle\qquad\qquad\qquad (k-1) J_N(\theta_{k-1})
-2k J_N(\theta_{k})\ \big).
\end{align*}
As for Prop.~\ref{prop:distrib-X}, the rest of the proof consists in expressing these results in terms of~$\kappa_k$.
\end{proof}

Notice that, from a simple change of variable, the value $\kappa_k$ can be conveniently rewritten as
\begin{align}
\kappa_k&=\textstyle\cN\,a\,(\sin\theta_k)^{N-1} -
k\,\cN\, (N-1)\int_0^{\theta_k} (\sin \theta)^{N-2}\,\ud\theta\nonumber\\
&=\textstyle\cN\,(N-1)\,\int_0^{\theta_k} (\sin \theta)^{N-2}\,(a\cos\theta -
k)\,\ud\theta\nonumber\\
\label{eq:kappa-alt-form}
&=\textstyle\cN a\, (N-1) \,\int_{0}^{1} (1 - u^2)^{\frac{N-3}{2}}\,(u -
\tfrac{k}{a})_+\,\ud u.
\end{align}

The following proposition bounds the moments of a random variable $X \sim \buf(a, N)$. These will be useful later for developing our
$\ell_2/\ell_1$ quantized embedding in Sec.~\ref{sec:quantized-embeddings}. 
\begin{proposition} 
\label{prop:buffon-moment-bounds}
Let $X \sim \buf(a,N)$. If $a<1$, for any $q\in\Nbb_0$, $\bb E X^q = \cN
  a$.\\
If $a\geq 1$, then $\bb E X^q \geq \cN a$ for any
$q\in\Nbb_0$. Moreover, for $a\geq 0$,
\begin{equation}
  \label{eq:buffon-2nd-moment-bound}
\max(\cN a, \tinv{N}a^2\,)\ \leq\ \bb E X^2\ \leq \cN a + \tinv{N}(a^2-1)_+, 
\end{equation}
and
\begin{equation}
  \label{eq:buffon-3rd-moment-bound}
\big|\bb E X^3 - \big(\cN a + \dN(\tfrac{3}{2})\,a^{3}\big)\big| \leq \tfrac{3}{N}\,a^{2}.
\end{equation}
For $q \geq 4$ and $a\geq 1$, the bounds are a bit more technical and read
\begin{align}
&\textstyle \big|\,\bb E X^q\ -\ \big(\cN a +
\dN(\tfrac{q}{2}) a^{q}
\big)\big|\nonumber\\[1mm]
&\label{eq:buffon-q-moment-bound}
\textstyle\leq\ q\,\dN(\frac{q-1}{2})\,a^{q-1} + \tfrac{1}{24}
{q \choose 2} \dN(\tfrac{q-2}{2}) (2a)^{q-2}\nonumber\\
&\textstyle\qquad\qquad + \tfrac{1}{12}
{q \choose 3}\,\dN(\frac{q-3}{2})\,(2a)^{q-3}.\qquad  
\end{align}
For any $q\geq 2$ and any $a\geq 0$, we have the upper bound
\begin{equation}
  \label{eq:weak-buffon-upper-bound}
  \textstyle \bb E X^q\ \leq\ \cN a + 2^{q-2}  \dN(\frac{q}{2})\,a^{q}\ +\ 2^{q-2}
q\,\dN(\frac{q-1}{2})\,a^{q-1}.
\end{equation}
\end{proposition}

This last proposition leads to a nice asymptotic relation.
\begin{corollary}
\label{cor:asympt-moment-bound}
For a Buffon random variable $X \sim \buf(a,N)$, we have
asymptotically in $a$, 
$$
\textstyle |\bb E X^q\ -\ \dN(\tfrac{q}{2}) a^q|  =
O(a^{q-1}). 
$$  
\end{corollary}

Before delving in the proof of Prop.~\ref{prop:buffon-moment-bounds},
we must introduce three useful lemmata.
\begin{lemma}
\label{lemma:sum-cj-pj}
For any sequence $\{c_k\}$
\begin{equation}
\label{eq:sum-cj-pj}
\sum_{k=0}^{n+1} c_k p_k\ =\ c_{0}(\kappa_{-1}-2\kappa_{0})+c_1\kappa_0 + \sum_{k=1}^{n} \Delta^2(c_{k-1}) \kappa_{k},  
\end{equation}
with the difference operator $\Delta$ such that $\Delta(c_k) = c_{k+1} - c_k$. 
\end{lemma}
\begin{proof}
Following \cite{diaconis1976buffon}, this is a simple consequence of the ``summing by parts'' rule for any sequences $a_k$ and $b_k$, \ie
$\sum_{k=0}^{n+1} a_j \Delta(b_j) = a_{n+2}b_{n+2} - a_0b_0 -
\sum_{k=0}^{n+1} \Delta(a_k) b_{k+1}$, and the fact that $
\sum_{k=0}^{n+1} c_k p_k
= \sum_{k=0}^{n+1} c_k \Delta^2(\kappa_{k-1})$.
\end{proof}

\begin{lemma}
  \label{lemma:D2-on-power-sequence}
  We can compute that $\Delta^2\big((k-1)^2\big) = 2$ and $\Delta^2\big((k-1)^3\big) = 6 k$, while for
  higher power $q\geq 4$ and $k\geq 1$, 
  \begin{equation}
    \label{eq:power-bound}
    \textstyle |\Delta^2\big((k-1)^q\big) - q(q-1) k^{q-2}|\ \leq\ 2 {q \choose 4} (2k)^{q-4}.
  \end{equation}
  A weaker bound reads
  \begin{equation}
    \label{eq:weak-power-bound}
    \textstyle \Delta^2\big((k-1)^q\big) \leq 2^{q-1} {q \choose 2} k^{q-2}.
  \end{equation}
\end{lemma}
\begin{proof}
The first two results come from the identities
$\Delta^2\big((k-1)^2\big) = (k+1)^2 + (k-1)^2 -
2k^2 = 2$ and $\Delta^2\big((k-1)^3\big) = (k+1)^3 + (k-1)^3 - 2k^3 = 6k$. The last one is obtained by
estimating $\Delta^2((k-1)^q)$ from
a third order Taylor development of both $(k+1)^q$ and
$(k-1)^q$ around $k$, their fourth order errors being both bounded by ${q \choose 4}
(k+1)^4 \leq {q \choose 4}
(2k)^4$.   The weaker bound is obtained similarly from a first order Taylor
development with a bounding of the second order error.
\end{proof}

\begin{lemma}
\label{lemma:kappa-moments-bounds}
The sum of $\kappa_k$ is bounded as 
\begin{equation}
  \label{eq:kappa-sum-bounds}
  \textstyle\tfrac{1}{N}\,a^2 - \cN a\ \leq\ 2\sum_{k=1}^{n} \kappa_k\
  \leq\ \tfrac{1}{N}\,(a^2 - 1)_+,
\end{equation}
while for other power $p\in \Nbb_0$,
\begin{equation}
\label{eq:lemma:kappa-moments-bounds-dN}
\textstyle\big|(p+1)(p+2)\sum_{k=1}^n k^{p}\kappa_k -
\dN(\frac{p}{2}+1)\,a^{p+2} \big|\ \leq (p+2)\ \dN(\frac{p+1}{2})\,a^{p+1}.
\end{equation}
\end{lemma}
\begin{proof}
Using the alternate formulation \eqref{eq:kappa-alt-form} of
$\kappa_k$, we find first for $p\geq 0$.
\begin{align}
\label{eq:power-kappa-sum-in-integral}
&\textstyle\sum_{k=1}^n k^{p}\kappa_k =\ \cN a\, (N-1) \,\int_{0}^{1} (1 - u^2)^{\frac{N-3}{2}}\,\sum_{k=1}^n k^{p}\,(u -
\tfrac{k}{a})_+\,\ud u. 
\end{align}
In the case where $p=0$, $\tinv{2}(u^2 - u)\ \leq\ \sum_{k=1}^{+\infty}\,(u -
k)_+\ \leq\ \tinv{2} u^2$. This is easily observed from $u = \lfloor u\rfloor + (u - \lfloor
u\rfloor) = \sum_{k=1}^{+\infty} \bb I(u \geq k) + (u - \lfloor
u\rfloor)$, which integrated gives $\tinv{2} u^2 =
\sum_{k=1}^{+\infty} (u - k)_+ + \int_0^u (v - \lfloor
v\rfloor)\,\ud v$, the last integral being a positive and smaller than
$\tinv{2}u$. Therefore, for any $a>0$, 
\begin{equation}
  \label{eq:sum-u-ka-low-up-bound}
\textstyle\tinv{2}(a u^2 - u)\ \leq\ \sum_{k=1}^{+\infty}\,(u -
\tfrac{k}{a})_+\ \leq\ \tinv{2} a u^2.  
\end{equation}
Moreover, for any $s\in \Nbb$ and given the definition of $\cN$, 
\begin{equation}
 \label{eq:other-def-of-dN}
\textstyle \cN (N-1) \int_{0}^{1} (1 -
u^2)^{\frac{N-3}{2}}\,u^{s}\,\ud u = \cN
\tfrac{N-1}{2}\,B(\tfrac{s+1}{2},\tfrac{N-1}{2})
= \tfrac{\Gamma\big(\tfrac{s+1}{2}\big)\Gamma\big(\tfrac{N}{2}\big)}{\sqrt{\pi}\,\Gamma\big(\tfrac{N+s}{2}\big)}
= \dN(\tfrac{s}{2}),
\end{equation}
with the ``Beta'' function $B(x,y) = \Gamma(x)\Gamma(y)/\Gamma(x+y)$
and $\dN$ defined in \eqref{eq:dN-def}.

Therefore, using \eqref{eq:power-kappa-sum-in-integral} combined with
the lower bound of \eqref{eq:sum-u-ka-low-up-bound} and the
identity $\Gamma(x)x=\Gamma(x+1)$ for any $x\in\Rbb_+$, we get
$$
\textstyle\sum_{k=1}^n \kappa_k \geq \ \tfrac{1}{2}\dN(1)\,a^2 -
\tinv{2}\dN(\frac{1}{2})a = \tfrac{1}{2N}\,a^2 - \tinv{2}\,\cN a.
$$ 
Similarly, the upper bound of \eqref{eq:sum-u-ka-low-up-bound} can lead to
$\textstyle \sum_{k=1}^n \kappa_k \leq \tinv{2} \dN(1) a^2 = \tfrac{1}{2N}\,a^2$.
A tighter bound is obtained by observing that, from \eqref{eq:power-kappa-sum-in-integral},
$\kappa_k(a=1) = 0$ and
\begin{align*}
\textstyle\sum_{k=1}^n \tfrac{\ud}{\ud
  a}\kappa_k(a)&\textstyle=\ \cN (N-1) \,\int_{0}^{1} (1 -
u^2)^{\frac{N-3}{2}} \sum_{k=1}^n u\,\bb I(u \geq
\tfrac{k}{a})\,\ud u\\
&\textstyle= \cN (N-1) \,\int_{0}^{1} (1 -
u^2)^{\frac{N-3}{2}}\,u\,\lfloor a u\rfloor\,\ud u\\
&\textstyle\leq \cN a (N-1) \,\int_{0}^{1} (1 -
u^2)^{\frac{N-3}{2}}\,u^2\,\ud u\\
&\textstyle =\ \tinv{N}\,a, 
\end{align*}
using \eqref{eq:other-def-of-dN} with $s=2$ in the last equality. Therefore, 
\begin{equation}
\textstyle \sum_{k=1}^n \kappa_k(a)\ =\ \bb I(a\geq 1)\,\int_1^a \sum_{k=1}^n \tfrac{\ud}{\ud
  u}\kappa_k(u)\, \ud u\ \leq \tinv{2N} (a^2 - 1)_+.  
\end{equation}
 
For analyzing positive power $p$, we rely on the fact that, for a continuous
and integrable function $g:[l,m]\to \Rbb$ with a unique extremum on 
$[l,m]\subset \Rbb$, 
$$
\textstyle\big|\,\sum_{k=l+1}^{m} g(k) - \int_l^{m} g(t)\,\ud t\big|\ \leq\ \max_{t\in[l,m]} |g(t)|.
$$ 
Taking $g(t)=t^p(u-t)$ which has a unique maximum on $\tfrac{p}{p+1}u$
of height $(\tfrac{p}{p+1})^{p}\tinv{p+1} u^{p+1} \leq \tinv{p+1}
u^{p+1}$, we find $
\big|\sum_{k=1}^{\infty} k^p(u-k)_+\ -\ \tinv{(p+2)(p+1)} u^{p+2}\big|\ \leq \tinv{p+1}
u^{p+1}$, 
since $\int_0^{\infty} t^p(u-t)_+\,\ud
t = u^{p+2} B(p+1,2) = \tinv{(p+2)(p+1)} u^{p+2}$. For any $a>0$, this
leads to 
$$
\bigg|(p+2)(p+1)\sum_{k=1}^{\infty} k^p(u-\tfrac{k}{a})_+\ -\ a^{p+1}u^{p+2}\bigg|\ \leq (p+2)
a^p u^{p+1}. 
$$
The result follows by inserting this last bound in
\eqref{eq:power-kappa-sum-in-integral} and reusing
\eqref{eq:other-def-of-dN} for $s\in\{p+1,p+2\}$.
\end{proof}

Notice that \eqref{eq:lemma:kappa-moments-bounds-dN} (in Lemma \ref{lemma:kappa-moments-bounds}) is probably
improvable for small values of $a$ since, as said in the proof above,
$\kappa_k(1)=0$. We note, however, that the expression is tight
asymptotically in~$a$. 

Thanks to the previous Lemmata, we are now ready to prove Prop.~\ref{prop:buffon-moment-bounds}.
\begin{proof}[Proof of Prop.~\ref{prop:buffon-moment-bounds}]
If $a<1$, then, for all
$q\geq 1$, $\bb E X^q = 1^q\,p_1 = \cN a$, while if $a\geq 1$,
\eqref{eq:sum-cj-pj} shows that $\bb E X^q = \kappa_0 + \sum_{k=1}^n
\Delta^2((k-1)^q) \kappa_n \geq \cN a$ since $\kappa_0 = \cN
a$.

Let us consider now more specific values of $q$ for the case $a\geq
1$. For $q=2$, we know from
Lemmata~\ref{lemma:sum-cj-pj} and \ref{lemma:D2-on-power-sequence}
that $\bb E X^2 = \kappa_0 + 2\sum_{k=1}^{n} \kappa_k$ and the upper bound
follows from \eqref{eq:kappa-sum-bounds} since $\kappa_0 = \cN
a$. 

For $q = 3$, the same two lemmata provide $\bb E X^3 = \kappa_0 +
6\sum_{k=1}^{n} k\kappa_k$. Moreover, from
\eqref{eq:lemma:kappa-moments-bounds-dN},
$$
\textstyle\big|6\sum_{k=1}^n k\kappa_k -
\dN(\frac{3}{2})\,a^{3} \big|\ \leq
3\,\dN(1)\,a^{2},
$$
which involves
$$
\big | \bb E X^3 - \big(\cN a + \dN(\tfrac{3}{2})\,a^{3}\big)\big|\ \leq\ 3\,\dN(1)\,a^{2}.
$$
For $q\geq 4$, the result becomes a bit technical. Again from Lemmata~\ref{lemma:sum-cj-pj} and \ref{lemma:D2-on-power-sequence}, 
$$
\big|\bb E X^q - \big(\kappa_0 +
q(q-1)\sum_{k=1}^{n} k^{q-2}\kappa_k\big)\big| \leq 2^{q-3}
{\textstyle {q \choose 4}} \sum_{k=1}^{n} k^{q-4} \kappa_k.
$$
Using twice \eqref{eq:lemma:kappa-moments-bounds-dN}, we find 
\begin{align*}
&\textstyle \big|\bb E X^q - \big(\kappa_0 +
\dN(\tfrac{q}{2}) a^{q}
\big)\big|\\
&\textstyle\leq\ q\,\dN(\frac{q-1}{2})\,a^{q-1} + 2^{q-3}{q \choose 4} \sum_{k=1}^{n} k^{q-4} \kappa_k\\
&\textstyle\leq\ q\,\dN(\frac{q-1}{2})\,a^{q-1} + \tfrac{2^{q-6}}{3}
q(q-1) \big( \dN(\tfrac{q-2}{2}) a^{q-2} + (q-2)\,\dN(\frac{q-3}{2})\,a^{q-3}\big)\\
&\textstyle\leq\ q\,\dN(\frac{q-1}{2})\,a^{q-1} + \tfrac{1}{24}
{q \choose 2} \dN(\tfrac{q-2}{2}) (2a)^{q-2} + \tfrac{1}{12}
{q \choose 3}\,\dN(\frac{q-3}{2})\,(2a)^{q-3}.
\end{align*}
Finally, for the weak upper bound \eqref{eq:weak-buffon-upper-bound},
we note that \eqref{eq:lemma:kappa-moments-bounds-dN} involves
$$
\textstyle {q \choose 2}\,\sum_{k=1}^n k^{q-2}\kappa_k \leq 
\inv{2} \dN(\frac{q}{2})\,a^{q}\ +\ \inv{2} q\,\dN(\frac{q-1}{2})\,a^{q-1}.
$$
Using \eqref{eq:sum-cj-pj} and \eqref{eq:weak-power-bound}, we obtain
\begin{align*}
\textstyle \bb E X^q&\textstyle\leq \cN a + 2^{q-1} {q \choose 2} \sum_{k=1}^n
k^{q-2}\\
&\textstyle \leq \cN a + 2^{q-2}  \dN(\frac{q}{2})\,a^{q} + 2^{q-2}
q\,\dN(\frac{q-1}{2})\,a^{q-1}.
\end{align*}

\end{proof}

\section{Quasi-Isometric Quantized Embedding}
\label{sec:quantized-embeddings}

Buffon's needle problem and its generalization to an $N$-dimensional
space lead to interesting observations in the field of dimensionality
reduction: it helps in understanding the
impact of quantization on the classical Johnson-Lindenstrauss (JL) Lemma~{\cite{johnson1984extensions,Achlioptas:2003p640}}. 

To see this, let us consider the common uniform quantizer of bin width $\qbin>0$
\begin{equation}
\label{eq:quantizer-def}
\textstyle \cl Q(\lambda) = \qbin \lfloor \frac{\lambda}{\qbin}
\rfloor\ \in\ \qbin\Zbb,
\end{equation}
defined componentwise when applied on vectors. Notice that we could
have defined the more common \emph{midrise} quantizer $\cl Q': \lambda \to \qbin \lfloor \lambda/\qbin
\rfloor + \delta/2$ with no impact on the rest of our developments.

Given a random matrix $\bs\Phi \sim \cl N^{M\times N}(0,1)$ and a
uniform random vector $\bs\xi \sim
\cl U^M([0, \qbin])$, we define the non-linear mapping $\bs \mapf_{\qbin}:\Rbb^N \to
\ZbbMqbin$ such that
\begin{equation}
\label{eq:quantiz-mappin-def-with-dithering}
\bs \mapf_{\qbin} (\bs u) = \cl Q(\bs\Phi \bs u + \bs\xi),
\end{equation}
where $\bs \xi$ plays a useful \emph{dithering} role: its action
randomizes the location of each unquantized component of $\bs
\Phi \bs u$ inside a quantization cell of $\Rbb^M$
\cite{gray1998quantization}. Our dithered construction is similar to the one
developed in \cite{Boufounos2010}, but our quantizer is different. 

How can we interpret the action of this mapping $\bs \mapf_{\qbin}$?
How does it approximately preserve the distance between a pair of points $\bs u, \bs v \in \Rbb^N$?
Surprisingly, the answer comes from Buffon's needle problem from the
following equivalence.

\begin{proposition} 
\label{prop:quant-embed-is-buffon}
Under the notations defined above, for each $j\in
  [M]$ and conditionally to the
  knowledge of $r_j = \|\bs \varphi_j\|$, we have 
  \begin{equation}
    \label{eq:quant-embed-is-buffon}
    \textstyle X_j := \inv{\qbin} |(\bs\mapf_\qbin(\bs u))_j -
    (\bs\mapf_\qbin (\bs v))_j |\ \sim_{\rm iid} \buf(\frac{r_j}{\qbin}
    \|\bs u -\bs v\|, N).    
  \end{equation}
\end{proposition}
\begin{proof}
Let $\cl
G$ be a grid of parallel $(N-1)$-dimensional hyperplanes that are $\qbin$
apart. Without any loss of generality, we assume them normal to the axis $\bs e_1= (1, 0, \cdots,
0)^T$ and each hyperplane corresponds to the set $\cl H_k = \{\bs
x\in\Rbb^N: \bs e_1^T \bs x = \qbin k\}$ for $k\in\Zbb$. Let us now imagine a
``needle'' $\needle(\bs u,\bs v)$ whose extremities are determined by two points $\bs u$ and
$\bs v$ somewhere in $\Rbb^N$. Note that the parameterization of the needle
with its extremities is
equivalent to the one defined in Sec.~\ref{sec:buff-needle-probl}. 

Notice that the number of
intersections $\needle(\bs u,\bs v)$ has with $\cl G = \cup_{k\in\Zbb}
\cl H_k$ can obviously be expressed with the quantizer $\cl Q$ as
$$
\textstyle \tinv{\qbin} |\cl Q(\bs e_1^T \bs u) - \cl Q(\bs e_1^T \bs v)|.
$$
The reason is that, if $\bs x\in\Rbb^N$ falls between $\cl H_{k(\bs x)}$ and $\cl H_{k(\bs x)+1}$ (the last
hyperplane excluded), then $\cl Q(\bs e_1^T \bs x) = k_{\bs x}\qbin$ with $k_{\bs x}:=\lfloor \frac{\bs e_1^T \bs x}{\qbin}
\rfloor$. Therefore, $\tinv{\qbin} |\cl Q(\bs e_1^T \bs u) - \cl Q(\bs
e_1^T \bs v)| = |k_{\bs u} - k_{\bs v}|$ is the number of hyperplanes
crossing $\needle(\bs u, \bs v)$.

Let us define now a random dithering $\xi \sim \cl U([0, \qbin])$ and a
random rotation $\gamma$ whose
distribution is uniform on the rotation group\footnote{This is made possible from the existence of a Haar measure
on $\SO(N)$ (see, \eg \cite{vershynin2011lectures}).} $\SO(N)$ of
$\Rbb^N$. From these, we can
create the mapping $\bs x_{\gamma, \xi} = T_{\gamma, \xi}(\bs x) = \bs R(\gamma)\,\bs x + \xi \bs
e_1$, where $\bs R(\gamma)\in\Rbb^{N\times N}$ stands for the matrix representation of $\gamma$. 

Thanks to this transformation, given two
vectors $\bs u,\bs v\in\Rbb^N$, the needle $\needle(\bs u_{\gamma,
  \xi},\bs v_{\gamma, \xi})$ of length $\|\bs u_{\gamma,
  \xi} - \bs v_{\gamma, \xi}\|=\|\bs u -\bs
v\|$ whose extremities are
defined by $\bs u_{\gamma, \xi}$ and $\bs v_{\gamma, \xi}$ is
oriented uniformly at random (conditionally to $\xi$) thanks to the action of $\gamma$,
\ie the random vector $\bs R(\gamma)(\bs u-\bs v)$ is uniform\footnote{\label{ft:SON-haar}
This is a simple consequence of the uniqueness of the Haar measure on
$\bb S^{N-1}$ and of the fact that, given any $\bs x\in\bb S^{N-1}$, $\bs R(\gamma)\bs x$ is
rotationally invariant if $\gamma$ is picked uniformly at random on $\SO(N)$.}
 on $\bb S^{N-1}$.

Moreover, conditionally to $\gamma$, this needle is also positioned uniformly
at random relatively to the
$\qbin$-periodic grid $\cl G=\cup_{k\in \Zbb} \cl H_k$. From the
action of the dithering, any fixed
point $\bs p \in \needle(\bs u_{\gamma,
  0},\bs v_{\gamma, 0})$ on the undithered needle (\eg its midpoint)
has an abscissa $p_1 + \xi \sim \cl U([p_1,p_1+\qbin])$ along $\bs e_1$ after
dithering. Therefore, from the periodicity of $\cl G$,
the distance between $\bs p + \xi \bs e_1 \in \needle(\bs u_{\gamma,
  \xi},\bs v_{\gamma, \xi})$ and the nearest hyperplane
of $\cl G$ is distributed as $\cl U([0, \qbin/2])$ conditionally to
$\gamma$.  

Consequently, the quantity
$$
\tinv{\qbin} |\cl Q(\bs e_1^T (\bs u_{\gamma, \xi})) - \cl
Q(\bs e_1^T (\bs v_{\gamma, \xi}))|
$$ 
counts the number of
intersections between $\cl G$ and the needle $\needle(\bs u_{\gamma,
  \xi},\bs v_{\gamma, \xi})$, which is oriented and positioned uniformly
at random relatively to $\cl G$. In other
words, we are in presence of a Buffon random variable ${\rm
  Buffon}(\|\bs u - \bs v\|/\qbin, N)$! 

Moreover, for any $\bs x\in\Rbb$, we have $\bs e_1^T \bs R(\gamma)
\bs x = (\bs R(\gamma)^{-1} \bs e_1)^T
\bs x \sim \bs \theta^T \bs x$
where $\bs \theta$ is a random vector uniformly distributed$^{\ref{ft:SON-haar}}$ on $\bb
S^{N-1}$. 
Therefore, 
\begin{equation}
\textstyle\!\!\!\! \tinv{\qbin} |\cl Q(\bs e_1^T (\bs u_{\gamma, \xi})) - \cl
Q(\bs e_1^T \bs v_{\gamma, \xi})|  \sim  \tinv{\qbin} |\cl Q(\bs
\theta^T \bs u + \xi) - \cl Q(\bs
\theta^T \bs v + \xi)|\ \sim\ \buf(\frac{1}{\qbin}\|\bs u - \bs v\|, N).
\end{equation}

Since any Gaussian random vector $\bs\varphi \sim \cl N^N(0,1)$ can be
written as $\bs \varphi = r \hat{\bs \varphi}$ with $r=\|\bs
\varphi\|$ and $\hat{\bs \varphi}=\bs\varphi/r$ picked uniformly at random on
$\bb S^{N-1}$, we can conclude that, conditionally to $r$, 
$$  
\textstyle\tinv{\qbin} |\cl Q(\bs
\varphi^T \bs u + \xi) - \cl Q(\bs
\varphi^T \bs v + \xi)| \sim \buf(\frac{r}{\qbin}\|\bs u - \bs v\|, N),
$$
which, from \eqref{eq:quantiz-mappin-def-with-dithering}, behaves exactly as the amplitude of one component of $\inv{\qbin} (\bs \mapf_{\qbin} (\bs u) - \bs \mapf_{\qbin}(\bs v))$.

Therefore, we can finally state that, for all $j\in[M]$ and conditionally
to the knowledge of the length $r_j = \|\bs \varphi_j\|$, 
$$
\textstyle X_j := \inv{\qbin} |(\bs\mapf_\qbin(\bs u))_j - (\bs\mapf_\qbin (\bs v))_j| \sim_{\rm iid} \buf(\frac{r_j}{\qbin}
\|\bs u -\bs v\|, N), 
$$
with the independence of the random variables $X_j$ resulting from the one of the rows of $\bs \Phi$.
\end{proof}

Now that this equivalence is proved, we see how to reach the
characterization of a quantized embedding determined by $\bs \mapf_{\qbin}$: we have
to study the concentration properties of each $X_j$ around their
mean. Therefore, targeting
the use of a classical concentration result due to Bernstein (explained later),
we first have to analyze the moments of these random variables.

Let us start with the evaluation of their expectation. Notice that, since $\bs \varphi_j
\sim_{\rm iid} \cl N^N(0, 1)$, each $r_j = \|\bs \varphi_j\| \sim_{\rm iid} \chi(N)$ follows a $\chi$ distribution
with $N$ degrees of freedom. We have also that for $Z \sim_{\rm iid} \chi(N)$
and $q\in\Nbb$, 
\begin{equation}
\label{eq:moment-norm-normal}
\bb E Z^q = 2^{\frac{q}{2}} \frac{\Gamma(\frac{N+q}{2})}{\Gamma(\frac{N}{2})} = \frac{2^{\frac{q}{2}}\Gamma(\frac{q+1}{2})}{\sqrt{\pi}\,\dN(\frac{q}{2})},
\end{equation} 
where $\dN$ was defined in \eqref{eq:dN-def}. 

This allows one to see that the expectation of each $X_j$ is
proportional to $\|\bs u - \bs v\|$ and independent of~$N$.
\begin{proposition}
\label{prop:expect-abs-dif-quant}
Let $\qbin>0$, $\bs\Phi \sim \cl N^{M\times N}(0,1)$, $\bs \xi\sim \cl U^M([0, \qbin])$ and $\cl Q$ defined as
above. Given $\bs u, \bs v\in\Rbb^N$ and $j\in [M]$, we have
\begin{equation}
  \label{eq:expectation-abs-dif-quant}
  \qbin\,\bb E X_j = \bb E\,|\cl Q(\bs\varphi_j^T \bs u +
  \xi_j) - \cl Q(\bs\varphi_j^T \bs v + \xi_j)|\ =\ \sqrt{\tfrac{2}{\pi}}\,\|\bs u - \bs v\|.
\end{equation}
\end{proposition}
\begin{proof}
The proposition follows from the law of total expectation applied to
the computation of $\bb E X_j$ with $X_j = \tinv{\qbin}\,|(\bs\mapf_\qbin (\bs u))_j
- (\bs\mapf_\qbin (\bs v))_j|$. Since, conditionally to $r=\|\bs\varphi_j\|$, $X_j
\sim \buf(\frac{r}{\qbin}\|\bs u - \bs v\|, N)$, and since $r
\sim \chi(N)$, we have
$$
\textstyle \bb E X_j = \bb E \big(\bb E(X_j|r) \big) = \cN\bb E(
\frac{r}{\qbin}\,\|\bs u - \bs v\|) = \sqrt{\frac{2}{\pi}}\, \frac{1}{\qbin}\,\|\bs u - \bs v\|.
$$  
\end{proof}

Beyond the mere evaluation of $\bb E X_j$, we can show that, if $\|\bs u - \bs v\|$ is much larger than
$\qbin$, any $X_j$ for $j\in [M]$
behaves like the amplitude of a Gaussian random variable $\cl N(0,\|\bs u - \bs
v\|^2/\qbin^2)$. This fact is established hereafter from an asymptotic analysis
of the moments $\bb E X_j^q$.
 
\begin{proposition} 
\label{cor:mapping-as-gaussian-asympt}
Following the previous conventions, for any
  $j\in[M]$ and $\alpha = \|\bs u - \bs v\|/\qbin$, we have 
$$
\big|\,\bb E X_j^q\ -\ \bb E|G_\alpha|^q\,\big|\ =\
O(\alpha^{q-1}),
$$
with $G_\alpha \sim \cl N(0, \alpha^2)$ and $~\bb E|G_\alpha|^q = \tinv{\sqrt
    \pi}\,2^{\frac{q}{2}}\,\alpha^q\,\Gamma(\frac{q+1}{2}) = O(\alpha^q)$.
\end{proposition}  
\begin{proof}
First notice that, if $Z\sim \chi(N)$, then, using \eqref{eq:moment-norm-normal} and a classical result on the absolute moments of a Gaussian random variable,
$$
\textstyle\dN(\frac{p}{2})\,\bb E Z^p = \tinv{\sqrt \pi}\,2^{\frac{p}{2}} \Gamma(\frac{p+1}{2})
= \bb E |G|^p,
$$ 
with $p\geq 0$ and $G\sim \cl N(0,1)$.  

Let us now consider the case $q\geq 4$. Therefore, considering the random mixture
$X_j\sim \buf(r_j\,\alpha, N)$ with $r_j = \|\bs \varphi_j\| \sim \chi(N)$, conditionally to $r_j$, \eqref{eq:buffon-q-moment-bound} provides 
\begin{align*}
&\textstyle \big|\,\bb E(X_j^q|r_j)\ -\ \big(\cN a +
\dN(\tfrac{q}{2}) a^{q}
\big)\big|\nonumber\\[1mm]
&\quad\textstyle\leq\ q\,\dN(\frac{q-1}{2})\,a^{q-1} + \tfrac{1}{24}
{q \choose 2} \dN(\tfrac{q-2}{2}) (2a)^{q-2} + \tfrac{1}{12}
{q \choose 3}\,\dN(\frac{q-3}{2})\,(2a)^{q-3}, 
\end{align*}
with $a = r_j\,\alpha$. From the law of total expectation, this shows that 
\begin{align*}
&\textstyle\big|\,\bb E X_j^q\ -\ \big(\bb E|G_\alpha| +
\bb E|G_\alpha|^q
\big)\big|\nonumber\\[1mm]
&\textstyle\leq q \bb E|G_\alpha|^{q-1} + \tfrac{2^{q-2}}{24}
{q \choose 2} \bb E|G_\alpha|^{q-2} + \tfrac{2^{q-3}}{12}
{q \choose 3} \bb E|G_\alpha|^{q-3}. 
\end{align*}
and the result follows since $\bb E|G_\alpha|^{p} = O(\alpha^p)$ for
any $p\geq 0$. The cases $1\leq q \leq 3$ are proved similarly from \eqref{eq:buffon-expect},
\eqref{eq:buffon-2nd-moment-bound} and \eqref{eq:buffon-3rd-moment-bound}.
\end{proof}

Corollary \ref{cor:mapping-as-gaussian-asympt} shows that, for $j\in [M]$, each random
variable $|(\bs\mapf_\qbin(\bs u))_j -
(\bs\mapf_\qbin (\bs v))_j |$ asymptotically behaves
like the amplitude of a Gaussian random variable of variance $\|\bs u - \bs v\|^2$ from the
proximity of their moments when this variance is large. Interestingly enough, without any
quantization, the random variable
$|(\bs \Phi \bs u)_j - (\bs \Phi \bs v)_j|$ exactly follows this
distribution for $\bs\Phi \sim \cl N^{M\times N}(0,1)$. Therefore, we
can expect later that the concentration properties of $\sum_j X_j$ should
converge to a Gaussian concentration behavior in the same asymptotic regime.

In parallel to this asymptotic analysis, bounds on the moments of $X_j$ can be estimated
thanks to those of a Buffon random variable, as summarized in
Prop.~\ref{prop:buffon-moment-bounds}.

\begin{proposition}
Let us
define $\alpha := \|\bs u - \bs v\|/\qbin$. In the conventions of Prop.~\ref{prop:expect-abs-dif-quant}, we have
\begin{equation}
  \label{eq:2nd-moment-abs-dif-quant}
\max(\sqrt{\tfrac{2}{\pi}}\,\alpha,\alpha^2) \leq \bb E X_j^2 \leq
\sqrt{\tfrac{2}{\pi}}\,\alpha + \alpha^2.
\end{equation}
and, for $q>2$,
\begin{equation}
\label{eq:q-moment-abs-dif-quant}
\textstyle \bb E X_j^q \leq \sqrt{\frac{2}{\pi}}\,\alpha +  \frac{2^{\frac{3}{2}q-2}}{\sqrt \pi}\, \alpha^q\,\Gamma(\frac{q+1}{2})+\frac{2^{\frac{3}{2} q- \frac{5}{2}}}{\sqrt \pi}\,\alpha^{q-1}\,q\,\Gamma(\frac{q}{2}).
\end{equation}
\end{proposition}

\begin{proof}
For the second moment, we start from
\eqref{eq:buffon-2nd-moment-bound} with $a=r_j\alpha$ and $r_j\sim
\chi(N)$ to get
\begin{equation}
\textstyle \bb E\max(\tinv{N}a^2, \cN a)\ \leq\ \bb E X_j^2 =\ \bb E \big(\bb E(X^2_j|r_j=\|\bs \varphi_j\|) \big)\ \leq\ \cN \bb E a + \tinv{N}\bb E (a^2-1)_+. 
\end{equation}
However, from \eqref{eq:moment-norm-normal},
\begin{equation}
  \label{eq:moment-rho-a}
  \textstyle \dN(\frac{q}{2})\,\bb E\, a^q = \frac{2^{\frac{q}{2}}}{\sqrt \pi\,\qbin^q}\, \|\bs u -
  \bs v\|^q\,\Gamma(\frac{q+1}{2}), 
\end{equation}
so that $\cN \bb E a = \sqrt{\frac{2}{\pi}}
\alpha$ and\footnote{Bounding $\bb E (a^2 - 1)_+$ more
  tightly is possible but this leads later to negligible improvements
  in our study.} $\tinv{N} \bb E (a^2 - 1)_+ \leq \tinv{N} \bb E a^2 = \alpha^2$ which leads to
$$
\textstyle \max \big ( \alpha^2, \sqrt{\frac{2}{\pi}}\,\alpha\big)\ \leq\ \bb E X_j^2 \ \leq\
\sqrt{\frac{2}{\pi}}\,\alpha + \alpha^2. 
$$
For higher moments, using $\bb E X_j^q = \bb E(\bb E (X_j^q|r_j))$ and
following the same techniques as above,
\eqref{eq:weak-buffon-upper-bound} and~\eqref{eq:moment-rho-a} provide the following upper bound
$$
\textstyle \textstyle \bb E X_j^q\ \leq\ \sqrt{\frac{2}{\pi}}\,\alpha +  \frac{2^{\frac{3}{2}q-2}}{\sqrt \pi}\, \alpha^q\,\Gamma(\frac{q+1}{2})
+\ \frac{2^{\frac{3}{2} q- \frac{5}{2}}}{\sqrt \pi}\,\alpha^{q-1}\,q\,\Gamma(\frac{q}{2}).
$$
\end{proof}

In the last proposition, we can also get rid of the $\Gamma$ functions by invoking the relation
$\Gamma(x+\tinv{2}) \leq \sqrt{x}\,\Gamma(x)$ \cite{wendel1948note}
whose recursive application provides $\Gamma(\frac{p+1}{2}) \leq
2^{\frac{1-p}{2}} \sqrt{\pi} \sqrt{p!}$. Using this we find,
for $q>2$,
\begin{equation}
  \label{eq:q-moment-abs-dif-quant-without-gamma}
\textstyle \bb E X_j^q\ \leq \sqrt{\frac{2}{\pi}}\,\alpha +  2^{q-\frac{3}{2}}\,\alpha^q\, \sqrt{q!}
+\ 2^{q- \frac{3}{2}}\,\alpha^{q-1}\,q\,\sqrt{(q-1)!}.
\end{equation}

Having delineated the behavior of the moments of each $X_j$, we can now study
their concentration properties. This is achieved from the Bernstein
inequality using a formulation from
\cite[p. 24]{massart2007concentration} that suits the rest of our developments.  
\begin{theorem}[Bernstein's inequality \cite{massart2007concentration}] 
\label{thm:bernstein}
Let $V_1,\cdots, V_M$ be
independent real valued random variables. Assume that there exist
some positive numbers $v$ and $\beta$ such that
\begin{equation}
  \label{eq:bernstein-cond-var}
  \textstyle\sum_{j=1}^M\ \bb E V_j^2 \leq v  
\end{equation}
and for all integers $k \geq 3$
\begin{equation}
  \label{eq:bernstein-cond-moment}
  \textstyle\sum_{j=1}^M\ \bb E V_j^k \leq \tinv{2}\,k!\,\beta^{k-2}\,v.  
\end{equation}
Then, for every positive $x$,
\begin{equation}
  \label{eq:bernstein-result}
\textstyle \bb P[\,\big|\sum_{j=1}^M\,(V_j - \bb E V_j) \big|\ \geq\ \sqrt{2 v x} + \beta x\,] \leq 2 e^{-x}.  
\end{equation}
\end{theorem}
Notice that setting $x=M\epsilon^2$ in
\eqref{eq:bernstein-result} with $\epsilon\geq 0$, we get:
\begin{equation}
  \label{eq:bernstein-result-with epsilon}
\textstyle \bb P\big [\,|\tinv{M}\sum_{j=1}^M (V_j - \bb E V_j)| \geq \sqrt{\frac{2}{M} v}\,\epsilon +
\beta \epsilon^2\,\big] \leq 2
e^{-\epsilon^2 M}.  
\end{equation}
This is the formulation that we use in the rest of this paper. From \eqref{eq:bernstein-result-with epsilon}, we must
focus our attention on the evolution of $\sqrt{2v/M}\,\epsilon +
\beta\epsilon^2$ once $v$ and $\beta$ are adjusted to the bounds
of $\bb E X_j^q$. From \eqref{eq:2nd-moment-abs-dif-quant}, we already know that 
\begin{equation}
  \label{eq:tmp-bound-sum-exj-sq}
\textstyle\sum_{j=1}^M \bb E X_j^2 \leq \ M\ \big(
\sqrt{\frac{2}{\pi}}\,\alpha + \alpha^2), 
\end{equation}
and from \eqref{eq:q-moment-abs-dif-quant-without-gamma} and for $q\geq 3$,
\begin{equation}
  \label{eq:tmp-bound-sum-exj-pow}
\textstyle \sum_{j=1}^M \bb E X_j^q \textstyle\leq\ 
\sqrt{\frac{2}{\pi}}\,M\alpha +  M\,(2^{q-\frac{3}{2}}\,\alpha^q\, \sqrt{q!} +\ 2^{q- \frac{3}{2}}\,\alpha^{q-1}\,q\,\sqrt{(q-1)!}).
\end{equation}
For simplifying our analysis, let us conveniently analyze two cases: a \emph{coarse quantization}
where $\alpha=\inv{\qbin}\|\bs u - \bs v\| < 1$ and a \emph{fine
  quantization} where $\alpha \geq 1$. Under \emph{coarse quantization} and for $q\geq
3$, \eqref{eq:tmp-bound-sum-exj-pow} provides
\begin{equation}
\textstyle \sum_{j=1}^M \bb E X_j^q \textstyle\leq\ 
\sqrt{\frac{2}{\pi}}\,M\alpha +  q!\,M\,2^{q-2}\,\alpha\,(1
+ \inv{\sqrt{3}}) \leq \tinv{2}\,q!\ 2^{q-2}\ M\,\alpha\,(\frac{\sqrt 2}{6\sqrt \pi} + 2(1 +
\inv{\sqrt 3})),
\end{equation}
while \eqref{eq:tmp-bound-sum-exj-sq} leads to
$$
\textstyle\sum_{j=1}^M \bb E X_j^2 \leq \ \big(
\sqrt{\frac{2}{\pi}} + 1) M \alpha < 2 M \alpha.
$$
Therefore, since $(\frac{\sqrt 2}{6\sqrt \pi} + 2(1 +
\inv{\sqrt 3})) < 4$, taking $v/M = 4$ and $\beta = 2$, we satisfy the
two Bernstein conditions\footnote{We could set
  $v/M = 4 \alpha$ but we found that this tighter choice complicates the
  presentation of the final
  concentration results.}. 
Under \emph{fine quantization} (\ie $\alpha \geq 1$),
\eqref{eq:tmp-bound-sum-exj-sq} gives now $\textstyle\sum_{j=1}^M \bb E X_j^2 \leq \ M\ (
\sqrt{\frac{2}{\pi}} + 1)\,\alpha^2$,
and, from \eqref{eq:q-moment-abs-dif-quant-without-gamma} and $q\geq
3$,
\begin{align*}
\textstyle \sum_{j=1}^M \bb E X_j^q&\textstyle\leq\ 
\sqrt{\frac{2}{\pi}}\,M\alpha +  M(2^{q-\frac{3}{2}}\,\alpha^q\,
                                     \sqrt{q!} +\ 2^{q- \frac{3}{2}}\,\alpha^{q-1}\,q\,\sqrt{(q-1)!})\\
&\textstyle\leq\ 
\sqrt{\frac{2}{\pi}}\,M\alpha^q +  M(2^{q-\frac{3}{2}}\,\alpha^q\,
  \sqrt{q!}+\ 2^{q- \frac{3}{2}}\,\alpha^{q}\,q\,\sqrt{(q-1)!})\\
&\textstyle\leq\ \tinv{2}\,q!\ (2\alpha)^{q-2}\ M\,\alpha^2\,(\frac{\sqrt 2}{6\sqrt \pi} + 2(1 +
\inv{\sqrt 3}))\\
&\textstyle < \tinv{2}\,q!\ (2\alpha)^{q-2}\ M\,(2\alpha)^2,
\end{align*}
We see that taking $v/M = 4\alpha^2$ and $\beta =
2\alpha$ is compatible with both inequalities. 
\medskip

Consequently, we can state that
$\sqrt{v/M} = O(1 + \alpha)$ and $\beta = O(1 + \alpha)$ around any
value of~$\alpha$. Therefore, if $0<\epsilon<\epsilon_0$ for some fixed
value $\epsilon_0>0$, 
\begin{equation}
  \label{eq:unique-upper-bound-2vm-p-beta}
  \exists\ c,c'>0\quad \text{such that}\quad \sqrt{2v/M}\epsilon +
  \beta\epsilon^2\ <\ (c + c' \alpha)\,\epsilon.  
\end{equation}

Let us cook now the
first important result concerning our
mapping~$\bs \mapf_{\qbin}$. 
\begin{proposition}
\label{prop:final-quant-embed-prop}
Fix $\epsilon_0>0$, $0< \epsilon\leq \epsilon_0$ and $\qbin>0$. There exist two values
$c,c'>0$ only depending on $\epsilon_0$ such that, for
$\bs\Phi \sim \cl N^{M\times N}(0,1)$ and $\bs \xi\sim \cl U^M([0, \qbin])$, both determining the mapping
$\bs\mapf_\qbin$ in \eqref{eq:quantiz-mappin-def-with-dithering}, and for $\bs u, \bs v\in\Rbb^N$, 
\begin{equation}
  \label{eq:expectation-abs-dif-quant}
\textstyle (1 - c\epsilon)\,\|\bs u - \bs v\|\,-\, c'\qbin\epsilon\ \leq\ \tfrac
{\sqrt \pi}{\sqrt 2 M}\,\|\bs\mapf_\qbin(\bs u) - \bs\mapf_\qbin(\bs
v)\|_1\ \leq\ 
(1 + c\epsilon)\|\bs u - \bs v\|\,+\, c'\qbin\epsilon.
\end{equation}
with probability
higher than $1 - 2 e^{-\epsilon^2M}$.
\end{proposition}
\begin{proof}
From \eqref{eq:unique-upper-bound-2vm-p-beta} and from Theorem
\ref{thm:bernstein}, we know that there exist two values
$c,c'>0$ such that 
$$
\textstyle \bb P\big [\,|\tinv{M}\sum_{j=1}^M (X_j - \bb E X_j)| \geq
(c + c' \alpha) \epsilon\,\big] \leq 2
e^{-\epsilon^2 M}.  
$$
Therefore, since 
\begin{equation}
\textstyle X_j = \tinv{\qbin}\,|(\bs\mapf_\qbin(\bs u))_j  -
(\bs\mapf_\qbin(\bs v))_j|\ = \tinv{\qbin}\,|\cl Q(\bs\varphi_j^T \bs u +
\xi_j) - \cl Q(\bs\varphi_j^T \bs v + \xi_j)|,
\end{equation}
with $\bb E X_j = \sqrt{\frac{2}{\pi}}\,\alpha$, we find 
$$
\textstyle 
\sqrt{\frac{2}{\pi}}(1 - c'\epsilon)
\alpha - c\epsilon 
\leq \tinv{M}\sum_{j=1}^M X_j \leq \sqrt{\frac{2}{\pi}}(1 + c'\epsilon)
\alpha + c\epsilon,
$$
with probability exceeding $1-2e^{-\epsilon^2M}$ which provides the result.  
\end{proof}

Finally, this last proposition provides the main result of this paper.
{
\newcounter{tmp}
\setcounter{tmp}{\value{proposition}}
\setcounter{proposition}{1}
\begin{proposition}
Let $\cl S \subset \Rbb^N$ be a set of $S$ points. Fix $0<\epsilon<1$
and $\qbin >0$. 
For $M > M_0 = O(\epsilon^{-2}\log S)$, there exist a non-linear mapping $\bs \mapf:\Rbb^N\to \ZbbMqbin$ and two
constants $c,c'>0$ such that, for all pairs $\bs u,\bs v\in \cl S$,
\begin{equation}
  \textstyle (1 - \epsilon)\,\|\bs u - \bs
  v\|\,-\,c\,\qbin\,\epsilon\ \leq\ \frac{c'}{M} \|\bs \mapf(\bs u) -
  \bs \mapf(\bs v)\|_1\ \leq\ 
(1 + \epsilon)\|\bs u - \bs v\|\,+\,c\,\qbin\,\epsilon.\tag{\ref{eq:quantiz-jl-lemma}}
\end{equation}  
\end{proposition}}
\setcounter{proposition}{\value{tmp}}

\begin{proof} 
The proof proceeds first by simplifying
\eqref{eq:expectation-abs-dif-quant} in
Prop.~\ref{prop:final-quant-embed-prop} 
with the change of variable
$c\epsilon\to\epsilon$ and with $\epsilon_0$ high enough so that
$0<\epsilon<1$ after this rescaling. Next, we follow the classical
proof of the Johnson-Lindenstrauss Lemma
\cite{johnson1984extensions,dasgupta99elementary} already sketched in the
Introduction. Given the mapping $\bs\mapf_\qbin$ associated to $\bs \Phi \sim \cl N^{M\times N}(0,1)$ and $\bs \xi$ through
\eqref{eq:quantiz-mappin-def-with-dithering}, and considering the ${S
  \choose 2}\leq S^2/2$ possible pairs of vectors
in $\cl S$, we apply a standard union
bound argument for jointly satisfying the inequality
\eqref{eq:bernstein-result-with epsilon} for all such pairs. If $M > M_0 = 2\epsilon^{-2} \log
S = O(\epsilon^{-2} \log
S)$, then $2\log S - \epsilon^2 M < 0$ and the global probability of
success is higher than $1 -
\exp(2\log S - \epsilon^2 M)>0$. 
Moreover, this probability can be
arbitrarily boosted close to 1 by repeating the random generation of
$\bs\mapf_\qbin$, considering then the event that at least one of the
generated mappings will satisfy \eqref{eq:quantiz-jl-lemma}. This
shows the existence of $\bs\mapf$ with probability 1, in the limit of
an increasingly large sequence of mappings.  
\end{proof}

\section{Towards an $\ell_2/\ell_2$ quantized embedding}
\label{sec:extensions}

We could ask ourselves if there exists another form of the quantized embedding given
in Prop.~\ref{prop:quantiz-jl-lemma}, one that involves only the use of
$\ell_2$-distances for both a set $\cl S \subset \Rbb^N$ and its image
in $\ZbbMqbin$. The expected asymptotic case would be obvious: in the
limit where $\qbin$ vanishes, the standard JL Lemma should be
recovered. 

Unfortunately, such an appealing result seems hard to reach with the
mathematical tools developed in this work. Instead, we are able to 
show the existence of a mapping $\bs \mapf$ that is ``close'' to this
situation in the sense that the $\ell_2$-distance in $\Rbb^N$ is
actually distorted by a non-linear function whose action is 
mostly perceptible when $\qbin$ is high with respect to the pairwise distance
of the embedded points. Noticeably, the additive distortion of the
mapping decays
also more slowly with $M$, \ie like $O((\log S/M)^{1/4})$, than for
the $\ell_2/\ell_1$ quasi-isometric mapping
of Prop.~\ref{prop:quantiz-jl-lemma}.

\begin{proposition}
\label{prop:quantiz-jl-lemma-l2-l2}
Let $\cl S \subset \Rbb^N$ be a set of $S=|\cl S|$ points and fix $0<\epsilon<1$. 
For $M > M_0 = O(\tinv{\epsilon^2}\log S)$, there exist a non-linear mapping $\bs \mapf:\Rbb^N\to \ZbbMqbin$ and one
constant $c>0$ such that, for all pairs $\bs u,\bs v\in \cl S$,
\begin{equation}
  \label{eq:quantiz-jl-lemma-l2-l2}
  \textstyle (1 - \epsilon)\,g_\qbin(\|\bs u - \bs
  v\|)\,-\,c\,\qbin\,\sqrt{\epsilon}\ \leq\ \frac{1}{\sqrt M} \|\bs
  \mapf(\bs u) - \bs \mapf(\bs v)\|\ \leq\ 
(1 + \epsilon)\,g_\qbin(\|\bs u - \bs
  v\|)\,+\,c\,\qbin\,\sqrt\epsilon,
\end{equation}  
for a certain non-linear function $g_\qbin(\lambda)$ such that $|g_\qbin(\lambda) - \lambda|=O(\sqrt{\qbin\lambda})$ for
$\lambda \gg \qbin$ and $|g_\qbin(\lambda) - (\sqrt
  2\lambda/\sqrt \pi)^{1/2}|=O(\lambda)$ for $\lambda < \qbin$.
\end{proposition}

For reasons that will become clear below, the function $g_\qbin$ is actually defined by 
\begin{equation}
  \label{eq:g-def}
\textstyle  g_\qbin(\lambda) := \qbin
g(\frac{\lambda}{\qbin}),\quad \textstyle g(\lambda) := (\bb E X^2_\lambda)^{1/2},  
\end{equation}
with the random mixture $X_\lambda \sim \buf(r\lambda,N)$ and
$r\sim \chi(N)$. Using
\eqref{eq:2nd-moment-abs-dif-quant}, we know that 
$\max(\sqrt{\tfrac{2}{\pi}}\lambda,\,\lambda^2)\ \leq\ g^2(\lambda)\ \leq\
\sqrt{\tfrac{2}{\pi}}\,\lambda + \lambda^2$, which provides the
asymptotic properties of $g_\qbin$ from
$$
\textstyle\max\big( (\sqrt{\tfrac{2}{\pi}}\qbin\lambda)^{1/2}, \lambda\big)\ \leq\ g_\qbin(\lambda)\ \leq\
(\sqrt{\tfrac{2}{\pi}}\qbin\lambda)^{1/2} + \lambda.
$$

Because of the action of $g_\qbin$, $\bs \mapf$ in Prop.~\ref{prop:quantiz-jl-lemma-l2-l2} does not
provide an $\ell_2/\ell_2$ quasi-isometric embedding of $\cl S$ in
$\ZbbMqbin$.  We are only close to this situation if the smallest pairwise distance $\nu_{\cl S}$ in $\cl
S$ defined in
\eqref{eq:nu-S-def} is large compared to $\qbin$.

Strictly speaking, we cannot even say that the mapping
$\bs \mapf$ in Prop.~\ref{prop:quantiz-jl-lemma-l2-l2} generates a
quasi-isometric embedding between $(\cl S, d_{\qbin})$ and
$(\bs\mapf(\cl S), \ell_2)$ with the function $d_{\qbin}(\bs u, \bs v)
= g_{\qbin}(\|\bs u - \bs v\|)$. Indeed, it is not sure if $d_\qbin$ is actually a
distance and, therefore, $(\cl S, d_{\qbin})$ is not a metric space,
which prevents us to match the basic requirements of
Def.~\ref{def:quasi-isom-def}. Nevertheless, the asymptotic behavior
of $g_\qbin$ shows that such a quasi-isometry is not far when the
pairwise distances between points of $\cl S$ are big compared to $\qbin$. 

However, we see that an ``almost''
$\ell_2/\ell_2$ quantized embedding exists between a finite set $\cl S\subset\Rbb^N$ and its
image in $\ZbbMqbin$ with multiplicative and additive embedding errors
decaying as $O(\sqrt{\log S/M})$ and $O((\log S/M)^{1/4})$,
respectively. This constitutes a striking difference with the $\ell_2/\ell_1$
quasi-isometric embedding of Prop.~\ref{prop:quantiz-jl-lemma} where
both kind of errors decay as $O(\sqrt{\log S/M})$. 

On a more practical side, we may be interested in using
Prop.~\ref{prop:quantiz-jl-lemma-l2-l2} for some numerical
applications. As explained in the
Prop.~\ref{prop:final-quant-embed-prop-l2-l2} at the end of this
section, a random construction of $\bs \mapf$ is simply provided by
\eqref{eq:quantiz-mappin-def-with-dithering} but unfortunately 
there is no known closed-form expression for
$g_\qbin$. We know only its quadratic and linear asymptotic behaviors
for large or small arguments, respectively. Despite the absence of
an explicit formula, it is probably possible to 
estimate numerically~$g_\qbin$ from \eqref{eq:g-def}. This could be done in two
steps. First, by integrating numerically
the second moment of a Buffon random variable ${\rm
  Buffon}(a,N)$ and fitting the result with a polynomial function of
$a$ with the desired level of accuracy in a certain range of
values. Second, since $a\sim \chi(N)$, by applying the law of total expectation to each term of this polynomial in $a$ using
\eqref{eq:moment-norm-normal}.

\medskip
Let us finish this section by proving
Prop.~\ref{prop:quantiz-jl-lemma-l2-l2}. The developments are quite similar to those presented in
Sec.~\ref{sec:quantized-embeddings}. They begin with the following result.
\begin{proposition}
\label{prop:final-quant-embed-prop-l2-l2}
Fix $\epsilon_0 >0$, $0< \epsilon\leq 1$ and $\qbin>0$. There exist two values
$c,c'>0$ only depending on $\epsilon_0$ such that, for $\bs\Phi \sim
\cl N^{M\times N}(0,1)$ and $\bs \xi\sim\cl U^M([0, \qbin])$ determining $\bs\mapf_{\qbin}$ in
\eqref{eq:quantiz-mappin-def-with-dithering}, and for $\bs u,\bs v\in\Rbb^N$,
\begin{multline}
  \label{eq:expectation-abs-dif-quant-l2l2}
\textstyle (1 - c\epsilon)\,g^2_\qbin(\|\bs u - \bs
v\|)\,-\,c'\qbin^2\epsilon \leq\ \tinv{M} \|\bs\mapf_{\qbin}(\bs u) -
\bs\mapf_{\qbin}(\bs v)\|^2\ \leq 
(1 + c\epsilon)\,g^2_\qbin(\|\bs u - \bs v\|)\, +\, c'\qbin^2\epsilon,
\end{multline}
with probability higher than $1 - 2 e^{-\epsilon^2M}$,
\end{proposition}

\begin{proof}
The proof requires to consider the moments of the random variable 
$\tilde X_j = X_j^{2}$ with $X_j$ defined by
\eqref{eq:quant-embed-is-buffon} and, as for Sec.~\ref{sec:quantized-embeddings}, to find reasonably small values for $v$
and $\beta$ for fulfilling \eqref{eq:bernstein-cond-var} and
\eqref{eq:bernstein-cond-moment} in Theorem~\ref{thm:bernstein} with
$V_j = \tilde X_j$. Notice that by definition of the function
$g$ above and by the equivalence \eqref{eq:quant-embed-is-buffon}, we have
\begin{equation}
  \label{eq:tilde-x-expec-sum}
\textstyle \tinv{M}\,\sum_{j=1}^M \bb E \tilde{X}_j = g^2(\alpha) =
\tinv{\qbin^2}\,g^2_\qbin(\|\bs u - \bs v\|),  
\end{equation}
for $\alpha = \|\bs u - \bs v\|/\qbin$. Moreover, \eqref{eq:2nd-moment-abs-dif-quant} provides
\begin{equation}
  \label{eq:g-qbin-bounds}
  \textstyle \max(\sqrt{\tfrac{2}{\pi}}\,\alpha,\alpha^2)\ \leq\ g^2(\alpha)\ \leq\
\sqrt{\tfrac{2}{\pi}}\,\alpha + \alpha^2,
\end{equation}
For the $q$-moments of $\tilde X_j$ with $q\geq 2$, we know
from~\eqref{eq:q-moment-abs-dif-quant} that
\begin{align}
\textstyle \bb E \tilde X_j^q&\textstyle\leq
\sqrt{\frac{2}{\pi}}\,\alpha +  \frac{2^{3q-2}}{\sqrt \pi}\, \alpha^{2q}\,\Gamma(q+\frac{1}{2})
+\ \frac{2^{3q- \frac{5}{2}}}{\sqrt \pi}\,\alpha^{2q-1}\,2q\,\Gamma(q)\nonumber\\
&\textstyle\leq
\sqrt{\frac{2}{\pi}}\,\alpha +  \frac{1}{2^{5/2}\sqrt{\pi}}\,
(2\sqrt 2\alpha)^{2q}\,q!
+\ \frac{1}{\sqrt \pi}\,(2\sqrt
2\alpha)^{2q-1}\,q!\nonumber\\
&\label{eq:mom-tilde-x-bound}
\textstyle =
\sqrt{\frac{2}{\pi}}\,\alpha +  \frac{q!}{2\sqrt{\pi}} (2\sqrt
2\alpha)^{2q-1}\,(\alpha+ 2),
\end{align}
using $\Gamma(q+\inv{2})\leq \sqrt q\, \Gamma(q) \leq q!/\sqrt 2$ for
$q\geq 2$. For coarse quantization, \ie $\alpha < 1$,
\eqref{eq:mom-tilde-x-bound} provides
\begin{align*}
\textstyle \bb E \tilde X_j^q&\textstyle\leq 
\sqrt{\frac{2}{\pi}}\,\alpha +  \frac{q!}{2\sqrt{\pi}} (2\sqrt
2)^{2q-4}(2\sqrt
2)^3\,3\alpha\\
&\textstyle\leq 
\sqrt{\frac{2}{\pi}}\,\alpha +  \tfrac{\sqrt 2}{\sqrt{\pi}} q!\, 8^{q-2}\,24\alpha\\
&\textstyle\leq \tinv{2}8^{q-2} q!\,\big(
1 + \,96\big) \frac{\sqrt 2}{2\sqrt \pi}\,\alpha < \tinv{2}8^{q-2} q!\,40\alpha
\end{align*}
Thus, we can select $v/M = 40$ and
$\beta = 8$. For fine
quantization and $\alpha > 1$, starting again from \eqref{eq:mom-tilde-x-bound}, a similar development provides 
\begin{align*}
\textstyle \bb E \tilde X_j^q&\textstyle\leq 
\sqrt{\frac{2}{\pi}}\,\alpha +  \frac{q!}{2\sqrt{\pi}} (2\sqrt
2\alpha)^{2q-4}(2\sqrt
2)^3\,3\alpha^4\\
&\textstyle\leq 
\frac{\sqrt 2}{\sqrt\pi}\,\alpha +  \tfrac{\sqrt 2}{\sqrt{\pi}} q!\, (8\alpha^2)^{q-2}\,24\alpha^4\\
&\textstyle\leq \tinv{2}(8\alpha^2)^{q-2} q!\,\big(
1 + \,96\big) \frac{\sqrt 2}{2\sqrt \pi}\,\alpha^4 < \tinv{2}(8\alpha^2)^{q-2} q!\,40\alpha^4,
\end{align*}
promoting the values $v/M = 40 \alpha^4$ and $\beta =
8\alpha^2$. 
\medskip

Consequently, gathering both quantization scenarios we have
$\sqrt{2v/M} = O(1 + \alpha^2)$ and $\beta = O(1 + \alpha^2)$ around any value of
$\alpha\geq 0$. Therefore, if $0<\epsilon<\epsilon_0$, there exist two
values $c, c'>0$ only depending on $\epsilon_0$ such
that 
$$
\sqrt{2v/M}\epsilon + \beta\epsilon^2\ \leq\ (c + c'\alpha^2)\epsilon.
$$
Applying Theorem~\ref{thm:bernstein} for this bound allows
one to state that
$$
\textstyle \bb P\big [\,|\tinv{M}\sum_{j=1}^M (\tilde X_j - \bb E
\tilde X_j)| \geq
(c + c' \alpha^2) \epsilon\,\big] \leq 2
e^{-\epsilon^2 M},  
$$
or equivalently, using \eqref{eq:tilde-x-expec-sum}, that 
$$
\textstyle \big|\frac{\qbin^2}{M}\sum_{j=1}^M \tilde X_j\ -\ g_\qbin(\|\bs u - \bs
v\|)\big|\ \leq\
(c\qbin^2 + c' \|\bs u - \bs v\|^2) \epsilon, 
$$ 
with probability exceeding $1 - e^{-\epsilon^2 M}$. Finally, using
\eqref{eq:g-qbin-bounds}, we see that with the same probability
\begin{equation*}
\textstyle (1 - c'\epsilon)\,g^2_\qbin(\|\bs u - \bs
v\|) - c\qbin^2\epsilon\ \leq\ \frac{\qbin^2}{M}\sum_{j=1}^M \tilde
X_j\ \leq\ (1 + c'\epsilon)\,g^2_\qbin(\|\bs u - \bs
v\|) + c\qbin^2\epsilon.  
\end{equation*}
\end{proof}

Given Prop.~\ref{prop:final-quant-embed-prop-l2-l2}, the proof
of Prop.~\ref{prop:quantiz-jl-lemma-l2-l2} is highly similar to the one of
Prop.~\ref{prop:quantiz-jl-lemma}. 

\begin{proof}[Proof of Prop.~\ref{prop:quantiz-jl-lemma-l2-l2}]
We first note that \eqref{eq:expectation-abs-dif-quant-l2l2} in
Prop.~\ref{prop:final-quant-embed-prop-l2-l2} is equivalent to
\begin{equation}
  \label{eq:expectation-abs-dif-quant-l2l2-alt}
\textstyle (1 - c\epsilon)\,g_\qbin(\|\bs u - \bs
v\|)\,-\,\qbin\sqrt{c'\epsilon} \leq\ \tinv{\sqrt M}
\|\bs\mapf_{\qbin}(\bs u) - \bs\mapf_{\qbin}(\bs v)\|\ \leq\ 
(1 + c\epsilon)\,g_\qbin(\|\bs u - \bs v\|)\, +\, \qbin\sqrt{c'\epsilon},
\end{equation}
using again the fact that $(a - b) \leq (a^2 - b^2)^{1/2}$ if $a>b>0$ and $(a^2+b^2)^{1/2} <
a+b$ for any $a,b>0$, and also the inequalities
$\sqrt{1-c\epsilon}\,\geq\,1-c\epsilon$ and $\sqrt{1+c\epsilon}
\,\leq\,1+c\epsilon$.  The rest of the proof is similar to the one of
Prop.~\ref{prop:quantiz-jl-lemma} in
Sec.~\ref{sec:quantized-embeddings} and we omit it for the sake of brevity.
\end{proof}

\section{Conclusion}
\label{sec:conclusion}

In this paper, we were interested in studying the behavior of the JL
Lemma when this one is combined with a uniform quantization
procedure of bin width $\qbin>0$. The main result of our study is the existence of a
(randomly constructed) $\ell_2/\ell_1$ quasi-isometric mapping between a set $\cl
S\subset \Rbb^M$ and $\ZbbMqbin$. Our proof relies on generalizing the well-known
Buffon's needle problem to an $N$-dimensional space, and in finding an
equivalence between this context and the quantization of randomly
projected pairs of points. The final observation of our analysis is
that such a mapping displays both an additive and a multiplicative
distortion of the pairwise distances of points in this set. The two
distortions vanish like $O(\sqrt{\log S/M})$ as the dimension
$M$ increases, while the additive distortion additionally scales like $\qbin$.  
As an aside, we have also obtained several interesting results concerning
the generalization of Buffon's needle problem in $N$ dimensions, delineating 
the behavior of the moments of the related random variable ${\rm
Buffon}(a, N)$. We have concluded our study by showing that there
exists a ``nearly''
$\ell_2/\ell_2$ embedding of $\cl
S\subset \Rbb^M$ in $\ZbbMqbin$ that displays a quasi-isometric
behavior. However, this mapping induces a non-linear distortion of the
$\ell_2$-distances in $\cl S$ and, compared to the $\ell_2/\ell_1$
embedding described above, the additive distortion decays more slowly as $O((\log S/M)^{1/4})$. 

We acknowledge the fact that there may exist other
quantization schemes (\eg non-regular) that, when combined with random
linear mappings, lead to faster distortion decays (\eg
exponential). For instance, in \cite{Boufounos2010} it is shown that
if two randomly projected vectors lead to equal quantized projections
according to a non-regular quantizer, \ie if their distance is 0 in
this projected domain, their true distance must decrease
exponentially with the projected space dimension $M$. The Locally
Sensitive Hashing (LSH) method introduced in~\cite{andoni2006near}
for reaching fast approximate nearest neighbors search is another form
of efficient quantized dimensionality reduction that approximately
preserves distances between embedded points.
Knowing if such results can be extended to provide quasi-isometric mappings with
faster distortion decays than $O(\sqrt{\log S/M})$ leads to interesting open questions.      

\section*{Acknowledgements}
\label{sec:acknowledgements}

The author thanks Valerio Cambareri (UCLouvain, Belgium) for his advices on the
writing of this paper. The author thanks also the anomynous reviewers
for their useful remarks for improving this paper, and one of them in
particular for having provided a short alternative proof of
Prop.~\ref{prop:quantiz-jl-lemma} (see
App.~\ref{sec:altern-proof-prop2}). Laurent Jacques is a Research Associate funded by the Belgian
F.R.S.-FNRS. 

\appendix

\section{Alternative proof for Prop.~\ref{prop:quantiz-jl-lemma}}
\label{sec:altern-proof-prop2}

During the reviewing process of this paper, an anonymous and expert
reviewer has provided an elegant and short alternative for the proof of
Prop.~\ref{prop:quantiz-jl-lemma}. This one relies on the properties
of \emph{sub-Gaussian} random distributions. We insert his/her
developments in this appendix as it contains powerful mathematical tools for
characterizing quantized random projections.

Before describing the proof, let us first provide a brief overview
of the properties respected by
sub-Gaussian random variables. The interested reader can consult~\cite{IntroNonAsRandom} for a comprehensive presentation of
these concepts and their implications in random matrix analysis. 

A random variable (\rv) $X$ is
sub-Gaussian if its sub-Gaussian norm\footnote{Also called Orlicz $\psi_2$ norm.} \cite{IntroNonAsRandom}
\begin{equation}
  \label{eq:sub-gaussian-def}
  \|X\|_{\psi_2}\ :=\ \sup_{p\geq 1}\ p^{-\frac{1}{2}} (\bb E |X|^p)^{\frac{1}{p}}  
\end{equation}
is finite. Examples of such \rv 's are Gaussian, Bernoulli, uniform or bounded
\rv's. In fact, in the Gaussian case, if $X \sim \cl N(0, \sigma^2)$,
then $\|X\|_{\psi_2} \leq c \sigma$ for some $c >0$ since, from
Stirling's formula, we get $\Gamma(x) = O(x^x)$ for $x > 1$ and $(\bb E|X|^p)^{1/p} =
(2^{p/2}\pi^{-1/2}\Gamma(\frac{p+1}{2}))^{1/p} = O(\sqrt p)$.  

Sub-Gaussian \rv's and their norm respect several interesting
properties. First, if $X$ is deterministic $\|X\|_{\psi_2} = |X|$. Since $\|\!\cdot\!\|_{\psi_2}$ is a
norm, given two sub-Gaussian \rv's $X$ and $Y$, $\|X\|_{\psi_2} = 0$ iff $X=0$, $\|\lambda X\|_{\psi_2} =
|\lambda| \|X\|_{\psi_2}$ for $\lambda \in \bb R$ and we have the triangle inequality $\|X +
Y\|_{\psi_2} \leq \|X\|_{\psi_2} + \|Y\|_{\psi_2}$. Moreover, from \eqref{eq:sub-gaussian-def},
\begin{equation}
  \label{eq:sg-norm-linf-bound}
  \textstyle \|X\|_{\psi_2} \leq \|X\|_{\infty} := \inf\{M\geq 0: \bb P(|X| \leq M)=1\},
\end{equation}
so any bounded \rv is necessarily sub-Gaussian. The sub-Gaussian
norm of a centered sub-Gaussian \rv is also easily bounded by
\begin{equation}
  \label{eq:sg-mean-dev}
  \textstyle \|X - \bb E X\|_{\psi_2} \leq
  \|X\|_{\psi_2} + \|\bb E X\|_{\psi_2} =  \|X\|_{\psi_2} + |E X| \leq
  \|X\|_{\psi_2} + E|X| \leq 2\|X\|_{\psi_2},
\end{equation}
where the second inequality uses Jensen's inequality.

In addition, sub-Gaussian \rv's have a tail bound characterized by their norm, \ie there exist two
constants $C,c>0$
such that for all $\epsilon\geq 0$,
\begin{equation}
  \label{eq:sg-tail-bound}
\textstyle   \bb P(|X|>\epsilon)\ \leq C\,e^{-c\,\epsilon^2/ \|X\|^2_{\psi_2}},  
\end{equation}
and \eqref{eq:sg-mean-dev} shows that for a smaller $c>0$, $\bb P(|X - \bb
E X|>\epsilon)\ \leq C\,e^{-c\,\epsilon^2/ \|X\|^2_{\psi_2}}$.

Finally, for any $D \in \bb N$ independent sub-Gaussian random variables $\{X_1,
\,\cdots, X_D\}$, their sum is approximately invariant under rotation,
which means
\begin{equation}
  \label{eq:approx-rot-inv}
  \textstyle \|\sum_i (X_i - \bb E X_i)\|^2_{\psi_2}\ \leq\ C \sum_i \|X_i - \bb E X_i\|^2_{\psi_2},  
\end{equation}
for some other constant $C>0$. 

\medskip

We are now ready to provide the announced alternative proof. Let us consider the dimension reduction map $\bs\mapf_\qbin(\bs x) := \cl Q_\qbin(\bs\Phi \bs x
  + \bs\xi)$ associated to our uniform quantizer
$\cl Q_\qbin(\cdot) := \qbin \lfloor \cdot / \qbin \rfloor$ (applied
componentwise) with step $\qbin>0$, to a random
Gaussian matrix $\bs\Phi
\sim \cl N^{M \times N}(0,1)$ and to a random dithering $\bs \xi \sim
\cl U^M([0,\qbin])$. Given two expressions $A$ and $B$, we also use below the simplified notation $A \lesssim
B$ (resp. $A \gtrsim B$) that means $A \leq c B$ (resp. $A \geq c B$) for some constant~$c>0$. 

\paragraph*{1. Concentration} Fix $\bs u, \bs v \in \cl S$, with $\cl
S \subset \bb R^N$ a
finite set of cardinality $S \in \bb N$. We can represent
\begin{equation*}
\textstyle \inv{M} \|\bs\mapf_\qbin (\bs u) - \bs\mapf_\qbin (\bs
v)\|_1 - \bb E \inv{M} \|\bs\mapf_\qbin (\bs u) - \bs\mapf_\qbin (\bs
v)\|_1 = \inv{M} \sum_{i=1}^M (Z_i - \bb E Z_i)  
\end{equation*}
where, for $i \in [M]$ and $\bs \varphi \sim \cl N^N(0,1)$,  
\begin{equation}
\label{eq:rev1-e1}
  Z_i \sim_{\rm iid} Z := |Q_\delta(\scp{\bs \varphi}{\bs u} + \xi) - Q_\delta(\scp{\bs \varphi}{\bs v} + \xi)|.
\end{equation}

Therefore, since for any $\bs x \in \bb R^N$, $\scp{\bs \varphi}{\bs
  x} \sim \cl N(0,\|\bs x\|^2)$ and $\|\scp{\bs \varphi}{\bs
  x}\|_{\psi_2} = \| |\scp{\bs \varphi}{\bs
  x}|\|_{\psi_2} \lesssim \|\bs x\|$, the random variable $Y :=
|\scp{\bs\varphi}{\bs u} - \scp{\bs \varphi}{\bs v}| = |\scp{\bs
  \varphi}{\bs u - \bs v}|$ satisfies 
\begin{equation}
  \label{eq:bound-Y-ZmY}
  \|Y\|_{\psi_2} \lesssim \|\bs u - \bs v\|_2,\quad \|Z-Y\|_\infty
  \leq 2\delta,  
\end{equation}
since $|\cl Q_\qbin(\lambda) - \lambda| \leq \qbin$ for all $\lambda
\in \bb R$. By the triangle inequality and using \eqref{eq:sg-norm-linf-bound}, we obtain
$$
\|Z\|_{\psi_2} \lesssim \|Y\|_{\psi_2} + \|Z-Y\|_{\psi_2}
\lesssim \|\bs u-\bs v\|_2 + \delta.
$$ 
Then, from \eqref{eq:sg-mean-dev} and \eqref{eq:approx-rot-inv},
$$
\textstyle \big\|\inv{\sqrt M}\sum_{i=1}^M (Z_i - \bb E Z_i)\big\|_{\psi_2}\
\lesssim\ \|\bs u - \bs v\|_2 + \delta.
$$
From \eqref{eq:sg-tail-bound} and by definition of sub-Gaussian norm,
this means also that, for some $c>0$, 
$$
\textstyle \bb P\big[\inv{\sqrt M}\sum_{i=1}^M (Z_i - \bb E Z_i) > t\big] \leq
2\exp\big(-\frac{ct^2}{(\|\bs u- \bs v\|_2 + \delta)^2}\big),
$$
for all $t > 0$. Choosing $t = \sqrt M \epsilon\,(\|\bs u - \bs v\|_2 + \delta)$, we
conclude that
\begin{equation}
 \label{eq:rev1-2}
\textstyle  \big| \inv{M} \|\bs\mapf_\qbin(\bs u) - \bs\mapf_\qbin
(\bs v)\|_1 - \bb E \inv{M} \|\bs\mapf_\qbin(\bs u) - \bs\mapf_\qbin
(\bs v)\|_1\big|\ \leq\ \epsilon \|\bs u - \bs v\|_2 + \epsilon \delta  
\end{equation}
with probability at least $1-2\exp(-c M^2\epsilon^2)$. For an
appropriate $C > 0$, if $M \geq C
\epsilon^{-2} \log S$ as in Prop.~\ref{prop:quantiz-jl-lemma}, then
the failure probability is smaller than $S^{−2}$. This allows us to
take a union bound over all pairs $\bs u, \bs v \in \cl S$ so that, with high
probability, \eqref{eq:rev1-2} holds simultaneously for all $\bs u,
\bs v \in S$.

\paragraph*{2. Expectation} It remains to show that the expectation in
\eqref{eq:rev1-2} is proportional (with a constant factor) to
$\|\bs u- \bs v\|_2$. Note that
\begin{equation}
  \label{eq:rev1-e3}
  \textstyle \bb E \inv{M} \|\bs\mapf_\qbin(\bs u) - \bs\mapf_\qbin(\bs v)\|_1 = \bb E Z,
\end{equation}
where $Z$ is defined in (\ref{eq:rev1-e1}). Moreover, 
\begin{equation}
  \label{eq:mean-EZ}
  \bb E Z = \sqrt{\tfrac{2}{\pi}}\,\|\bs u - \bs v\|,
\end{equation}
as established in Prop.~\ref{prop:expect-abs-dif-quant}. This last
result can also be derived in a simpler fashion by observing that
for any $\bs\varphi \sim \cl N^{N}(0,1)$, 
$\scp{\bs \varphi}{\bs w} = \scp{\bs \varphi}{\cl P \bs w} = \scp{\cl
  P \bs \varphi}{\cl P \bs w}$ for all $\bs w \in \cl W:={\rm span}(\bs u,\bs v)$, with $\cl P$
the orthogonal projection on $\cl W$. Since this last space is a
two-dimensional subspace and since $\cl
  P \bs \varphi$ is distributed as $\cl N^2(0,1)$ (by rotation
  invariance), an easy variant of Prop.~\ref{prop:expect-abs-dif-quant}
  in 2-D based on Prop.~\ref{prop:expect-intersect} (borrowed from
  \cite{ramaley1969buffon}), \ie without generalizing Buffon's
  needle problem in $N$-D, suffices to prove \eqref{eq:mean-EZ}. Injecting
  \eqref{eq:mean-EZ} in 
  \eqref{eq:rev1-2} establishes finally Prop.~\ref{prop:quantiz-jl-lemma}.
\medskip

\noindent\emph{Remark:} All the developments above remain true for random matrices with rows selected
uniformly at random over $\sqrt N\,\bb S^{N-1}$, \ie when they are
i.i.d.~as ${\rm Unif}(\sqrt N\,\bb S^{N-1})$. In this case, those rows are also
sub-Gaussian random vectors and \eqref{eq:bound-Y-ZmY} also holds~\cite{vershynin2011lectures}. The only difference
lies in the mean of $\bb E Z$ in \eqref{eq:mean-EZ} with $Z$ defined
as in \eqref{eq:rev1-e1} for $\bs \varphi
\sim {\rm Unif}(\sqrt N\,\bb S^{N-1})$. In this case, $\bb E Z \neq
\sqrt{(2/\pi)}\,\|\bs u-\bs v\|$ but $|\|\bs u-\bs v\|^{-1}\bb E Z -
\sqrt{(2/\pi)}|=O(1/\sqrt N)$. Indeed, by rotation
invariance and since $\bb E_{u}(|\lfloor a +
u\rfloor - \lfloor b + u\rfloor|) = |a-b|$ for $a,b\in \bb R$ and $u
\sim \cl U([0,1])$, developing $Z$ from its definition
in \eqref{eq:rev1-e1} and using the
law of total expectation, we find 
$$
\bb E Z = \bb E_{\bs \varphi} \bb E_{\xi} Z = \bb E_{\bs \varphi}
|\scp{\bs \varphi}{\bs u - \bs v}| = \|\bs u - \bs v\|\,\bb E_{\bs \varphi}
|\varphi_1|.
$$

The pdf of $|\varphi_1|/\sqrt N$ is known (see, \eg \cite{stam1982limit}) and
reads $f(z) = (N-1)\tau_N\,(1-z^2)^{\tfrac{d-3}{2}}$ with
$\tau_N$ defined in \eqref{eq:buffon-expect}. Therefore, using \eqref{eq:other-def-of-dN}
\begin{equation*}
\textstyle \bb E |\varphi_1| = \int_0^1 z f(z)\,\ud z = (N-1)\tau_N\,\int_0^1
(1-z^2)^{\tfrac{d-3}{2}} z \ud z = \dN(\tfrac{1}{2}) =
\tau_N,  
\end{equation*}
with $\dN$ defined in \eqref{eq:dN-def}. Consequently, 
$$
\textstyle \bb E Z = \sqrt N \tau_N \|\bs u - \bs v\| = \frac{\sqrt
  N\,\Gamma(\frac{N}{2})}{\sqrt \pi \Gamma(\frac{N+1}{2})}\,\|\bs u - \bs v\|.
$$
Since $(\frac{2N-3}{4})^{1/2} \leq
{\Gamma(\frac{N}{2})}/{\Gamma(\frac{N-1}{2})} \leq
(\frac{N-1}{2})^{1/2}$, we have also 
$$
\tfrac{\sqrt{2}}{\sqrt \pi}\,\tfrac{\sqrt N\sqrt{N-\frac{3}{2}}}{(N-1)} \leq 
\sqrt N \tau_N 
= \tfrac{2\sqrt N}{\sqrt \pi (N-1)}\tfrac{\Gamma(\frac{N}{2})}{\Gamma(\frac{N-1}{2})} 
\leq \tfrac{\sqrt{2}}{\sqrt \pi}\tfrac{\sqrt N}{\sqrt{N-1}},
$$
so that $|\sqrt N \tau_N - {\sqrt{2}}/{\sqrt \pi}| = O(1/\sqrt
N)$. Therefore, Prop.~\ref{prop:quantiz-jl-lemma} holds also for
random matrices $\bs \Phi$ with rows i.i.d.~as ${\rm Unif}(\sqrt N\,\bb S^{N-1})$.

\end{document}